%% file: Final.tex
\documentclass[11pt,journal,onecolumn]{IEEEtran}
\usepackage{setspace}
\include{preamble}

\begin{document}

	\title{\vspace{5.5mm}   Three Variants of Differential Privacy: Lossless Conversion and Applications\thanks{This work was supported in part by NSF under grants CIF 1922971, 1815361, 1742836, 1900750, and CIF CAREER 1845852. Part of the results in this paper was presented at the International Symposium on Information Theory 2020 \cite{HundredRounds}.}}

\author{%
Shahab Asoodeh, Jiachun Liao, Flavio P. Calmon, Oliver Kosut, and Lalitha Sankar\thanks{S. Asoodeh and F. P. Calmon are with School of Engineering and Applied Science, Harvard University (e-mails: \{shahab, flavio\}@seas.harvard.edu). J. Liao, O. Kosut, and L. Sankar are with School of Electrical, Computer, and Energy Engineering at Arizona State University (e-mails: \{jiachun.liao, okosut, lalithasankar\}@asu.edu) }
}	
	\date{}
	\maketitle

% \begin{abstract}
% We investigate the problem of optimally
% characterizing differential privacy parameters of a mechanism satisfying a certain level of R\'enyi differential privacy
% %converting the privacy guarantee of R\'enyi differential privacy  to that of approximate differential privacy 
% using the joint range of $f$-divergences. As an application, we adopt the \textit{moments accountant} framework by Abadi et al. and show that our results may lead to about 100 more iterations in training deep learning models using stochastic gradient descent within the same privacy budget compared to the state-of-the-art results.      
% \end{abstract}

% Abstract suggestion"
\begin{abstract}
    We consider three different variants of differential privacy (DP), namely approximate DP,  R\'enyi DP (RDP), and hypothesis test DP. In the first part, we develop a machinery for optimally relating approximate DP to RDP based on the joint range of two $f$-divergences that underlie the approximate DP and RDP. In particular, this enables us to derive  the  optimal approximate DP parameters  of  a  mechanism  that  satisfies  a  given  level  of  RDP. As an application, we apply our result to the moments accountant framework for characterizing privacy guarantees of noisy stochastic gradient descent (SGD). When compared to the state-of-the-art, our bounds may lead to about 100 more stochastic gradient descent iterations for training deep learning models for the same privacy budget. In the second part, we establish a relationship between RDP and hypothesis test DP which allows us to translate the RDP constraint into a tradeoff between type I and type II error probabilities of a certain binary hypothesis test.  We then demonstrate that for noisy SGD our result leads to tighter privacy guarantees compared to the recently proposed $f$-DP framework for some range of parameters.
\end{abstract}
% \SA{
% \begin{itemize}
%     \item The conversion rule specified in Theorem 1 dictates that $\eps>\frac{1}{\alpha}$. So the small $\eps$ pushes the RDP to pure DP (as it causes $\alpha\to \infty$)! This theorem is hence inefficient for high privayc regime.
% \end{itemize}}
% We derive the optimal approximate differential privacy (DP) parameters of a mechanism that satisfies a given level of Rényi differential privacy (RDP). Our result is based on the joint range of two $f$-divergences that underlie the approximate and the Rényi variations of differential privacy. We apply our result to the moments accountant framework for characterizing privacy guarantees of stochastic gradient descent. When compared to the state-of-the-art, our bounds may lead to about 100 more stochastic gradient descent iterations for training deep learning models for the same privacy budget. 

\section{Introduction}

Differential privacy (DP) \cite{Dwork_Calibration} has become the \emph{de facto} standard for privacy-preserving data analytics. Intuitively, a randomized algorithm is said to be \emph{differentially private} if its output does not vary significantly with small perturbations of the input.  DP guarantees are usually cast in terms of properties of the difference of the \emph{information density} \cite{Pinsker1964InformationAI} of the algorithm's output and two different inputs---referred to as the \textit{privacy loss} random variable in the DP literature.  % DP enables  population-wide statistical analysis while guaranteeing a quantifiable measure of  individual privacy [ADD REFs]. 
In fact, several variants of DP has been proposed based on different properties of privacy loss random variable. Informally speaking, a mechanism is said to satisfy $(\eps, \delta)$-DP \cite{Dwork_Calibration} if the privacy loss random variable is bounded by $\eps$ with probability $1-\delta$. A mechanism is said to be $(\alpha, \gamma)$-\emph{R\'enyi differential privacy} (RDP) \cite{RenyiDP} if the $\alpha$th moment of the privacy loss random variable is upper bounded by $\gamma$; see Sec.~\ref{Sec:Preliminary} for more details.

Several methods have recently been proposed to ensure differentially private training of machine learning (ML) models \cite{Abadi_MomentAccountant, Shokri:DeepLearning, chaudhuri2011differentially, Bassily1, Balle:Subsampling, SGD_Exponential_VS_Gaussian}. Here, the parameters of the model determined by a learning algorithm (e.g.,  weights of a neural network or coefficients of a regression) are sought to be differentially private with respect to the data used for fitting the model (i.e. the  \emph{training} data).  When the model parameters are computed by applying stochastic gradient descent (SGD) to minimize  a given  loss function, DP can be ensured by directly adding noise to the gradient. %evaluated over samples drawn from the training data \cite{Bassily1, chaudhuri2011differentially, Shokri:DeepLearning, Abadi_MomentAccountant}. 
The empirical and theoretical flexibility of this  noise-adding procedure for ensuring DP  was demonstrated, for example, in \cite{Shokri:DeepLearning, Abadi_MomentAccountant}. This method is currently being used for privacy-preserving training of large-scale ML models in industry, see e.g., the implementation of \cite{McMahan_Privacy} in the Google's open-source TensorFlow Privacy
framework \cite{TensorFlowPrivacy}.

Not surprisingly,  for a fixed training dataset, privacy deteriorates with each SGD iteration. In practice, the DP  constraints (i.e., $\eps$ and $\delta$) are set \emph{a priori}, and then mapped to a permissible number of  SGD  iterations for fitting the model parameters. Thus, a key question is: \emph{given a DP constraint, how many  iterations are allowed before the SGD algorithm is no longer private?} 
The main challenge in determining the DP guarantees provided by  noisy  SGD is keeping track of the evolution of the privacy loss random variable during subsequent gradient descent iterations.  %
This can be done, for example, by invoking  advanced composition theorems for DP, such as  \cite{Vadhan_Composition, Kairouz_Composition}. Such composition results, while theoretically significant,  may be loose due to their generality (e.g., they do not take into account the noise distribution used by the privacy mechanism). %For example, the particular noise distribution used in a privacy mechanism. 

Recently, Abadi et al.\ \cite{Abadi_MomentAccountant} circumvented the use of DP composition results by developing a method called \textit{moments accountant} (MA). Instead of dealing with DP directly, the MA approach provides privacy guarantees in terms of RDP for which  composition has a simple linear form \cite{RenyiDP}. Once the privacy guarantees of the SGD execution are determined in terms of RDP, they are mapped back to DP guarantees in terms of $\eps$ and $\delta$ via a relationship between DP and RDP \cite[Theorem 2]{Abadi_MomentAccountant} allowing for converting from one to another. This approach renders tighter DP guarantees than those obtained from  advanced composition theorems (see \cite[Figure 2]{Abadi_MomentAccountant}).
Nevertheless, the existing conversion rules between RDP and DP are loose. In this work, we provide a framework which settles the \textit{optimal} conversion between RDP and DP, and thus further enhances the privacy guarantee obtained by the MA approach.   
Our technique relies on the information-theoretic study of joint range of $f$-divergences: we first describe both DP and RDP using two certain types of the $f$-divergences, namely $\sE_\lambda$ and $\chi^\alpha$ divergences (see Section~\ref{Sec:Preliminary}). We then apply \cite[Theorem 8]{Harremoes_JOintRange} to characterize the joint range of these two $f$-divergences which, in turn, leads to the ``optimal'' conversion between RDP and DP (see Section~\ref{Sec:Conversion}).  
Specifically, this optimal conversion allows us to derive bounds on the number of noisy SGD iterations for a given DP parameters $\eps$ and $\delta$. Our result improves upon the state-of-the-art \cite{Abadi_MomentAccountant} by allowing more training iterations (often hundreds more) for the same privacy budget, and thus providing higher utility for free (see Section~\ref{Sec:Gaussian}).

In the second part of this work, we revisit another variant of DP based on binary hypothesis testing. Consider an attacker who, given a mechanism's output, aims to determine if a certain individual (say Alice) has participated in the input dataset. This goal can be thought of as a hypothesis testing problem: rejecting the null hypothesis corresponds to the absence of Alice in the input dataset.  
It is well-known that $(\eps, \delta)$-DP is equivalent to enforcing that the type II error probability of \textit{any} (possibly randomized) such test at significance level (or type I error probability) $\tau$ is lower bounded by $1-\delta - e^\eps \tau$ \cite{Wasserman,Kairouz_Composition}. Thus, for small $\eps$ and $\delta$, any test is 
essentially powerless, i.e., it is impossible to have both small type I and type II error probabilities.  
This view of privacy (which we henceforth call \textit{hypothesis test DP}) brings an operational interpretation for DP. This notion of privacy has recently been parameterized by a convex and decreasing function $f:[0,1]\to [0,1]$ that specifies the tradeoff between type I and type II error probabilities. A mechanism is said to be $f$-DP \cite{GaussianDP} if, given a mechanism's output, the type II error probability of any test for a given significance level $\tau$ is lower bounded by $f(\tau)$. Thus, if $f(\tau)$ is approximately $1-\tau$, then any tests will be essentially powerless.  This new definition is shown to provide easily interpretable privacy guarantees. This is in sharp contrast with RDP whose privacy guarantee does not enjoy a clear interpretation (see \cite{Balle2019HypothesisTI} for more details). 

Our goal is to address the interpretability issue of RDP by relating RDP to $f$-DP. We first prove an explicit expression for the RDP guarantee of mechanism in terms of the  type I and type II probabilities corresponding to the ``optimal'' test (given by Neyman-Pearson lemma). We remark that our expression is similar to an unproved formula that appeared first in \cite[Eq. (2.79)]{Polyanskiy_thesis}. Conversely, we develop a machinery to implicitly relate RDP constraint to $f$-DP by constructing an achievable region of type I and type II error probabilities among all tests. This relationship is in particular interesting for the privacy analysis of iterative ML algorithm in that it converts the simple linear composition property of RDP to an interpretable privacy guarantee in terms of $f$-DP. Another approach for deriving an interpretable and tight privacy guarantee for ML algorithms is to resort to the general composition result of $f$-DP \cite[Theore 3.2]{GaussianDP}. This approach is advocated in \cite{GaussianDP_deep} for the privacy analysis of noisy SGD in training neural networks. We compare our results with \cite{GaussianDP, GaussianDP_deep} in two different directions: 

\begin{itemize}
    \item  The $f$-DP guarantee can be easily related to $(\eps, \delta)$-DP (see \cite[Proposition 2.12]{GaussianDP}). It is argued in \cite[Theorems 1 and 2]{GaussianDP_deep} that $f$-DP guarantee of SGD \textit{always} yields a strictly stronger $(\eps, \delta)$-DP guarantee than what would be obtained by moments accountant.  We empirically show that this does not hold if one incorporates our optimal RDP-to-DP conversion rule into the moments accountant framework; i.e., the \textit{improved} moments accountant might outperform $f$-DP, see Fig.~\ref{fig:MA_GDP}. 
    \item Rather then using the general composition results of $f$-DP, we propose to apply the linear composability of RDP and then convert the resulting guarantee to $f$-DP. Focusing on SGD with Gaussian noise, we demonstrate that there exists a threshold for variance below which our approach strictly outperforms  $f$-DP, see Fig.~\ref{fig:RDP_Tradeoff_SGD_sigmaT} and Fig.~\ref{fig:RDP_Tradeoff_SGD}. 
\end{itemize}

%Fig.~\ref{fig:diagram} summarizes our results on relationships between these three variants of DP. 

\subsection{Related Work}
Since the introduction of the approximate DP in \cite{Dwork_Calibration}, it has been extensively studied especially for iterative ML algorithms, see \cite{SGD_Exponential_VS_Gaussian, chaudhuri2011differentially, Chaudhuri_Subsampling, BAssily2, BAssily_NIPS19, Chaudhuri_logistic, DP_OnlineLearning,Thakurta_LASSO,Sarwate_SGD_Update,duchi2013local,Jain_risk_bound, Smith_interaction,Talwar_Privacte_LASSO,Private_ERM_General} to name a few. 
Perhaps one of the most fundamental
primitive in statistical privacy is the study of composition; how privacy degrades under as the algorithm iterates. 
There are still continued efforts to better understand the composition of DP. The advanced composition result for DP was derived \cite{Vadhan_Composition}. In a pioneering work, \cite{Kairouz_Composition} obtained an optimal homogeneous composition theorem for $(\eps, \delta)$-DP. It is, however, shown to be  $\#$P hard to compute the DP parameters under heterogeneous composition \cite{Vadhan_Murtagh}. A substantial recent effort has been devoted to relaxing the DP constraints using divergences between probability distributions to address the weakness of $(\eps, \delta)$-DP in handling
composition \cite{Abadi_MomentAccountant, Concentrated_Dwork,ZeroDP, RenyiDP,Bun:tCDP,Balle_AnalyticMomentAccountant}. 
For instance, \cite{Abadi_MomentAccountant, Concentrated_Dwork,ZeroDP, RenyiDP} considered R\'enyi divergence and showed that the optimal privacy parameters under composition have simple linear forms. Once composition is handled, the resulting privacy parameters are converted to $(\eps, \delta)$-DP via some conversion rule, e.g.,  \cite[Theorem 2]{Abadi_MomentAccountant}, \cite[Proposition 1.3]{ZeroDP}, and \cite[Proposition 3]{RenyiDP}. This technique significantly improves on earlier privacy analysis of SGD.
This technique has been extended by follow-up work \cite{McMahan_Privacy}. 
More recently, a new relaxed version of DP (not divergence-based), termed $f$-DP was proposed in \cite{GaussianDP} and shown to to enjoy a rather simple composition property. This new definition of DP was used in \cite{GaussianDP_deep} for the privacy analysis in training deep neural networks.

\subsection{Paper Organization} 
In Section~\ref{Sec:Preliminary}, we provide several preliminary definitions and results and also mathematically formulate our main goals. In Section~\ref{Sec:Conversion}, we characterize the optimal relationship between RDP and DP and apply it to the moments accountant framework in Section~\ref{Sec:Gaussian}. The content of Sections~\ref{Sec:Conversion} and \ref{Sec:Gaussian} appeared in the conference version \cite{HundredRounds} without proofs.  Section~\ref{Section:BHT} concerns the second main goal of the paper, that is, deriving a relationship between RDP and hypothesis test DP.         
 
\subsection{Notation}
We denote by $\mathbbD$ the universe of all possible datasets and by $(\mathbbX, \calF)$ a measurable space with Borel $\sigma$-algebra $\calF$. We also use $\mathcal P(\mathbb X)$ to denote the set of all probability measures on $\mathbb X$. 
We use capital letters, e.g., $X$ to denote random variables. We write $X\sim P$ to describe the fact that $X$ is distributed according to $P$. We also use the notation $\sim$ to indicate the neighborhood relationship between datasets, i.e., given two datasets $d$ and $d'$, we write $d\sim d'$ if their Hamming distance is equal to one.
For a pair of distributions $P$ and $Q$ and constant $\alpha\geq 1$, we let \begin{equation}
    D_\alpha(P\|Q)\coloneqq \frac{1}{\alpha-1}\log\mathbb E_Q\Big[\big(\frac{\text{d}P}{\text{d}Q}\big)^\alpha\Big]
\end{equation} denote the R\'enyi divergence of order $\alpha$. Also, given a real-valued convex function $f$ satisfying $f(1)=0$, the $f$-divergence \cite{Csiszar67, Ali1966AGC} between $P$ and $Q$ is defined as 
\begin{align}\label{eq:fdivergence}
		D_f(P\|Q) \coloneqq \mathbb \mathbb \mathbb E_Q\Big[f\big(\frac{\text{d}P}{\text{d}Q}\big)\Big].
	\end{align}
For any real number $a$, we write $(a)_+$ for $\max\{a, 0\}$ and for $a\in [0,1]$, we write $\bar a$ for $1-a$.

\section{Preliminaries and Problem Setup}\label{Sec:Preliminary}
In this section, we revisit several definitions and basic results that will be key for the discussion in the subsequent sections.  A mechanism $\calM: \mathbbD\to \calP(\mathbbX)$ assigns a probability distribution $\calM_d$ to each dataset $d\in \mathbbD$. Given a pair of neighboring datasets $d\sim d'$, the privacy loss random variable is defined as $L_{d,d'}\coloneqq \log\frac{\text{d}\calM_d}{\text{d}\calM_{d'}}(X)$ where $X\sim \calM_d$ and $\frac{\text{d}\calM_d}{\text{d}\calM_{d'}}$ represents the Radon-Nikodym derivative. 
Given an output of the mechanism $\calM$ and a pair $d\sim d'$,  consider the following problem of testing $\calM_d$ against $\calM_{d'}$: 
\begin{align}\label{Hypothesis}
    H_0:~ X \sim \calM_d~~~\text{vs.}~~~
    H_1:~ X \sim \calM_{d'}.
\end{align}
Let $\beta^{dd'}_\calM:[0,1]\to [0,1]$ denote the \textit{optimal} tradeoff between type I error  (i.e., the probability of declaring $H_1$ when the truth is $H_0$) and type II error  (i.e., the probability of declaring $H_0$ when the truth is $H_1$). More specifically,  $\beta^{dd'}_\calM(\tau)$ is the smallest type II error when type I error  equals $\tau$. The mapping $\tau\mapsto \beta^{dd'}_\calM(\tau)$ is sometimes called the \textit{tradeoff function}.

\begin{defn}\label{Def_RDP}
A mechanism $\calM:\mathbbD\to \calP(\mathbbX)$ is said to be
\begin{itemize}
     \item $(\eps, \delta)$-DP \cite{Dwork_Calibration} for $\eps\geq 0$ and  $\delta\in [0,1)$ if 
    \eqn{DP}{\sup_{A\in \calF, d\sim d'}\calM_{d}(A)-e^\eps\calM_{d'}(A)\leq \delta.}
    \item $(\alpha, \gamma)$-RDP \cite{RenyiDP} for $\alpha>1$ and $\gamma\geq 0$ if 
    \eqn{RDP}{\sup_{d\sim d'}D_\alpha(\calM_{d}\|\calM_{d'})\leq \gamma.}
    \item $f$-DP \cite{GaussianDP} for a convex and non-increasing function\footnote{Both $f$-DP and $f$-divergence are defined in terms of convex functions. In order to be consistent with their original notation, we use $f$ to denote the function in both definitions. It will be clear from the context and as a result will not lead to confusion.} $f:[0,1]\to [0,1]$ if  for all $\tau\in [0,1]$
     \begin{equation}\label{Def:f_DP}
         \inf_{d\sim d'}\beta_\calM^{dd'}(\tau)\geq f(\tau).
     \end{equation}
\end{itemize}
\end{defn}

\begin{figure}[t]
    \centering
    \includegraphics[scale = 0.4]{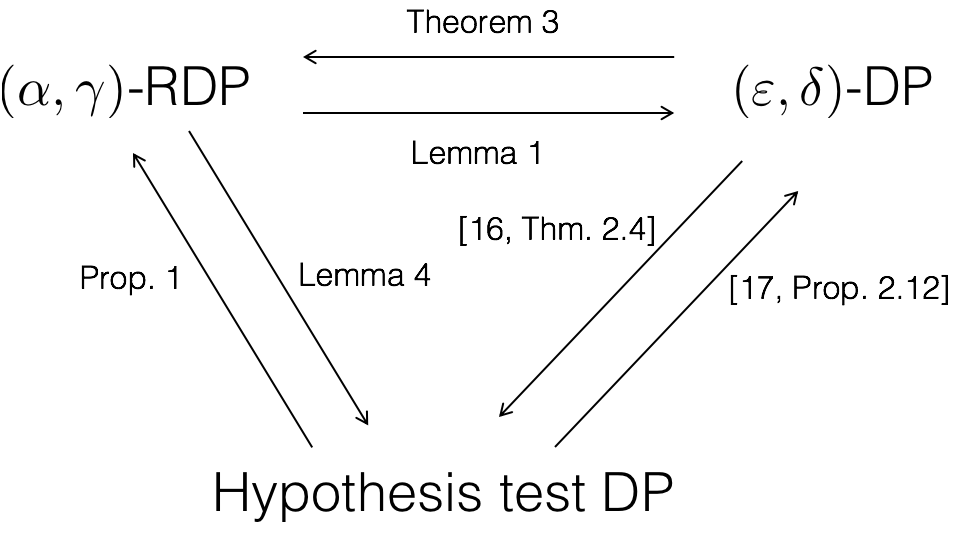}
    \caption{The diagrammatic summary of the key relationships between three variants of DP studied in the paper. }
    \label{fig:diagram}
\end{figure}
 It can be shown that \eqref{DP} is implied if the tail event $\{L_{d,d'}>\eps\}$ occurs with probability at most $\delta$ for all $d\sim d'$,  and \eqref{RDP} is implied if (and only if) the $\alpha$th moment of $L_{d,d'}$ is upper bounded by $\gamma$. 
 %Built on this intuition, the MA restricts the $\alpha$-moment of $L_{d,d'}$ for \textit{all} $\alpha>1$. 
 It is worth noting that the definition of RDP is closely related to \textit{zero-concentrated} DP \cite{ZeroDP, Concentrated_Dwork}. 
Different properties of these two variants of DP have been extensively studied. One well-studied property of these two definitions is the composition (to be discussed in details in Section~\ref{Sec:Gaussian}). As mentioned earlier, RDP tightly handles composition as opposed to the existing composition theorems for $(\eps, \delta)$-DP \cite{Vadhan_Composition, Kairouz_Composition} known to be either loose for many practical mechanisms or intractable to compute \cite{Vadhan_Murtagh}.  With this clear advantage comes a shortcoming: RDP suffers from the lack of operational interpretation, see e.g., \cite{Balle2019HypothesisTI}. 
%It is recently shown in \cite{Balle2019HypothesisTI} that, unlike DP, RDP cannot be equivalently expressed in terms of a constraint in binary hypothesis testing.
To address this issue, the RDP guarantee is often translated into a DP guarantee via the following result. 
%proved in \cite[Proposition 3]{RenyiDP} (see also \cite[Theorem 2]{Abadi_MomentAccountant}).
\begin{thm}[{{\cite[Thm 2]{Abadi_MomentAccountant}}}]
\label{Theorem:RDP_Delta}
If the mechanism $\calM$ is $(\alpha, \gamma)$-RDP, then it satisfies $(\eps, \delta)$-DP for any $\eps>\gamma$ and 
\begin{equation}\label{Eq:Delta1}
    \delta = e^{-(\alpha-1)(\eps- \gamma)}.
    % \delta = \inf_{\alpha>1}e^{-(\alpha-1)(\eps- \gamma)}.
    % {\color{blue}\eps = \gamma \frac{1}{\alpha-1}\log\frac{1}{\delta}. (1\geq \delta>0)}
\end{equation}

\end{thm} 
% This conversion rule is loose in general and does not hold for all range of $\eps\geq 0$. Nevertheless, as we shall see in Sec.~\ref{Sec:Gaussian}, this conversion rule is the main building block of the \textit{moments accountant} (MA) \cite{Abadi_MomentAccountant} which is the current state-of-the-art privacy analysis technique for ML algorithms.  
This theorem establishes a relationship between RDP and DP that is extensively used in several recent differentially private ML applications, e.g., \cite{Balle2019mixing, Balle:Subsampling, RDP_SemiSupervised, RDP2_PosteriorSampling, Feldman2018PrivacyAB, Balle_AnalyticMomentAccountant, bhowmick2018protection} to name a few. A prime use case for this relationship is the \textit{moments accountant} (MA) \cite{Abadi_MomentAccountant} which is the current state-of-the-art privacy analysis technique for ML algorithms.
However, despite its extensive use, Theorem~\ref{Theorem:RDP_Delta} is loose in general and does not hold for all range of $\eps\geq 0$. For instance, as we see later, for Gaussian mechanisms this relationship holds for $\eps\to 0$ only when the variance of noise goes to infinity. Given its widespread applications, it seems very natural to  to aim at determining the \textit{optimal} relationship between $(\eps, \delta)$-DP and $(\alpha, \gamma)$-RDP.   More precisely, we seek to answer the following question.

 \noindent\textbf{Question One:} \textit{Given an $(\alpha, \gamma)$-RDP mechanism $\calM$, what are the smallest $\eps$ and $\delta$ such that $\calM$ is $(\eps, \delta)$-DP?}

We settle this question in  Sec.~\ref{Sec:Conversion} by expressing the optimal relationship via a simple one-variable convex optimization program. Incorporating this  relationship into MA, we introduce the \textit{improved} MA and quantify the resulting improvement in terms of privacy and utility in two different settings: $T$-fold homogeneous composition of Gaussian mechanism and  noisy SGD algorithm.

As we shall see later, RDP is remarkably efficient in handling composition, making it an appealing notion of privacy for iterative algorithms such as SGD. Nevertheless, it lacks interpretability, see, e.g., \cite{Balle2019HypothesisTI} for more details. Following the success of hypothesis test \eqref{Hypothesis} in providing interpretation of approximate DP, we seek to relate RDP constraint to the tradeoff function $\inf_{d\sim d'}\beta^{dd'}_\calM(\tau)$. 
Such relationship enables us to provide interpretable privacy guarantees for several iterative machine learning algorithms. It can be verified that $1-\tau - \inf_{d\sim d'}\beta^{dd'}_\calM(\tau)$ quantifies the fundamental indistinguishability of  neighboring datasets based on the mechanism's output. Therefore, one effective way to describe the above relationship is to construct an outer bound for the region encompassed between the curves $\tau\mapsto 1-\tau$ and $\tau\mapsto \inf_{d\sim d'}\beta^{dd'}_\calM(\tau)$; the so-called \textit{privacy region} of $\calM$. Now we can describe the above relationship as follows. 

\noindent\textbf{Question Two:} \textit{Given an $(\alpha, \gamma)$-RDP mechanism, what is the characterization of its privacy region?}

We provide an outer bound for the solution of this question in Sec.~\ref{Section:BHT} and then demonstrate it for both $T$-fold homogeneous composition of Gaussian mechanism and  noisy SGD algorithm. Interestingly, for the latter scenario, this outer bound is tighter than what would be obtained from applying results in \cite{GaussianDP_deep}.

We summarize our results on the relationship between $(\eps, \delta)$-DP, $(\alpha, \gamma)$-RDP, and $f$-DP in Fig.\ref{fig:diagram}.

\section{Optimal Relationship between RDP and DP}
\label{Sec:Conversion}
In this section, we aim at computing the fundamental worst-case DP privacy parameter guaranteed by an $(\alpha, \gamma)$-RDP mechanism, thereby answering Question One. To this goal,
we first express constraints in both $(\eps, \delta)$-DP and $(\alpha, \gamma)$-RDP in terms of two $f$-divergences.  Given $\lambda\geq 1$, the $f$-divergence associated with $f(t) = (t-\lambda)_+ =\max\{t-\lambda, 0\}$, is called $\sE_{\lambda}$-divergence \cite{polyanskiy2010channel} (aka  \textit{hockey-stick} divergence \cite{hockey_stick}) and given by  
	\begin{equation}
	\label{Defi_HS_Divergence}
		\sE_{\lambda}(P\|Q) = \int (\text{d}P-\lambda \text{d}Q)_+=\sup_{A\in \calF} \left[P(A) - \lambda Q(A)\right].
	\end{equation}
Also, for any $\alpha>1$, the $f$-divergence associated with $f(t)=\frac{1}{\alpha-1}(t^\alpha-1)$ is denoted by\footnote{$\chi^\alpha$-divergence is also referred to as $\alpha$-Hellinger divergence, see, e.g., \cite{fdivergence_Sason16}.} $\chi^\alpha(P\|Q)$. Note that $D_\alpha(P\|Q) = \frac{1}{\alpha-1}\log\left(1 + (\alpha-1)\chi^\alpha(P\|Q)\right)$ for a pair of probability distributions $P$ and $Q$. 
% \begin{align}\label{eq:chi-alpha-divergence}
% 	\chi^\alpha(P\|Q) = \frac{1}{\alpha-1} \left(\int \left(\frac{\text{d}P}{\text{d}Q}\right)^\alpha\text{d}Q -1\right).
% \end{align}

It is shown in \cite{Barthe:2013_Beyond_DP} that 
\begin{equation}\label{DP_HS}
    \calM~\text{is}~(\eps, \delta)\text{-DP}~~ \Longleftrightarrow ~~\sup_{d\sim d'}\sE_{e^\eps}(\calM_d\|\calM_{d'})\leq \delta.
\end{equation}
Similarly, it can be verified that:
\begin{equation}\label{RDP-chi-alpha}
    \calM~\text{is}~(\alpha, \gamma)\text{-RDP} \Longleftrightarrow\sup_{d\sim d'}\chi^\alpha(\calM_d\|\calM_{d'})\leq \chi(\gamma),
\end{equation}
where 
\begin{equation}\label{tilde_Gamma}
   \chi(\gamma) \coloneqq \frac{e^{(\alpha-1)\gamma}-1}{\alpha-1}.
\end{equation}
Let the set of all $(\alpha, \gamma)$-mechanisms be denoted by $\mathbb M_\alpha(\gamma)$, i.e., 
$$\mathbb M_\alpha(\gamma)\coloneqq \{\calM:\mathbbD\to \calP(\mathbbX):~\calM~\text{is}~(\alpha, \gamma)\text{-RDP}\}.$$
This definition, together with \eqref{DP_HS}, enables us to precisely formulate Question One. In fact, Question One amounts to computing $\delta_\alpha^\eps(\gamma)$
\begin{align}
    \delta_\alpha^\eps(\gamma) & \coloneqq  \inf\left\{\delta\in (0,1)
    :\forall \calM\in \mathbb M_\alpha(\gamma)~\text{is}~ (\eps, \delta)\text{-DP}\right\}\\
    &= \sup_{\calM\in \mathbb M_\alpha(\gamma)}~ \sup_{d\sim d'}~\sE_{e^\eps}(\calM_d\|\calM_{d'}),
    \label{Fundamen_Smallest_Delta}
\end{align}
where the equality comes from \eqref{DP_HS} and \eqref{RDP-chi-alpha}. 
The map $\gamma\mapsto \delta_\alpha^\eps(\gamma)$ in fact specifies the ``optimal'' conversion rule from RDP to DP for a given $\eps\geq 0$. An equivalent way of describing such conversion is through the following quantity fixing $\delta\in (0,1)$
\begin{equation}\label{FundamentalEps}
    \eps_\alpha^\delta(\gamma)\coloneqq \inf\left\{\eps\geq 0:\forall \calM\in \mathbb M_\alpha(\gamma)~\text{is}~ (\eps, \delta)\text{-DP}\right\}.
\end{equation}
Similarly, the optimal conversion from DP to RDP is formulated by 
\begin{align}
        \gamma_\alpha^\eps(\delta)&\coloneqq \sup\left\{\gamma\geq 0 :\forall \calM\in \mathbb M_\alpha(\gamma)~\text{is}~ (\eps, \delta)\text{-DP}\right\}\label{FundamentalGamma}\\
        &= \inf_{\calM:\mathbbD\to \calP(\mathbbX)}\inf_{d\sim d'} \chi^{-1}(\chi^\alpha(\calM_{d}\|\calM_{d'}))\label{eq:Optimization_Gamma_DP}\\
	&\qquad \text{s.t.~} \sEs(\calM_{d}\|\calM_{d'})\geq \delta,~ \forall d\sim d',
\end{align}
where the $\chi^{-1}(\cdot)$ denotes the functional inverse of $\chi(\cdot)$ in \eqref{tilde_Gamma}, i.e., and is given by 
$\chi^{-1}(t) =  \frac{1}{\alpha-1}\log(1+(\alpha-1)t)$. We seek to compute $\delta_\alpha^\eps(\gamma)$ (or equivalently, $\eps_\alpha^\delta(\gamma)$); however, it turns out that $\gamma_\alpha^\eps(\delta)$ is simpler to compute. As a result, in the following we focus on the latter first.     

Notice that, according to \eqref{RDP-chi-alpha},  the set $\mathbb M_\alpha(\gamma)$ can be equivalently characterized by the constraint $\chi^\alpha(\calM_d\|\calM_{d'})\leq \chi(\gamma)$, where $\chi(\gamma)$ is defined in \eqref{tilde_Gamma}. Hence, $\gamma\mapsto \delta_\alpha^\eps(\gamma)$ in fact constitutes the upper boundary of the convex set 
\eqn{ConvexSet1}{\calR_\alpha\coloneqq \left\{\left(\chi^\alpha(\calM_{d}\|\calM_{d'}), \sEs(\calM_{d}\|\calM_{d'})\right)\Big|\forall\calM, d\sim d'\right\}.}This simple observation has some key implications.
First, $\delta_\alpha^\eps(\cdot)$ is non-decreasing and concave. Second, the upper boundary can be equivalently given by the map $\delta\mapsto \gamma_\alpha^\eps(\delta)$. Furthermore, to compute  $\gamma_\alpha^\eps(\cdot)$ or $\delta_\alpha^\eps(\cdot)$, it suffices to characterize $\calR_\alpha$. This  allows us to cast the problem of computing $\gamma_\alpha^\eps(\cdot)$ as characterizing the joint range of $\sE_\lambda$ and $\chi^\alpha$ divergences. To tackle the latter problem, we refer to \cite{Harremoes_JOintRange} whose main result is as follows. 
 \begin{thm}(\cite[Theorem 8]{Harremoes_JOintRange})\label{thm:f-divergence_joint-region}
	We have  
	\begin{equation}
		\Big\{\big(D_{f}(P\|Q), D_{g}(P\|Q)\big)\Big | P, Q\in\calP(\mathbbX) \Big\}		= {\sf conv}(\calB)
	\end{equation}
where $\conv(\cdot)$ denotes the convex hull operator and  
\begin{equation*}
    \calB\coloneqq \Big\{\big(D_{f}(P_{\sf b}\|Q_{\sf b}), D_{g}(P_{\sf b}\|Q_{\sf b})\big)\Big | P_{\sf b}, Q_{\sf b}\in\calP(\{0, 1\}) \Big\}.
\end{equation*}
 \end{thm}
This theorem states that characterizing the joint range of any pair of $f$-divergences can be reduced without loss of generality to the binary case. For completeness, we give a more direct proof for the case of $\chi^\alpha$ and $\sE_\lambda$ divergences in Appendix~\ref{Appexdix:binarycase}. 
We formalize this insight in Theorem~\ref{thm:optimization_formulation_gamma} and establish a simple variational formula for $\gamma_\alpha^\eps(\cdot)$ involving a one-parameter log-convex minimization program.   
Hence,  the optimization \eqref{eq:Optimization_Gamma_DP}, which can potentially be of significant complexity, turns into a simple tractable problem. 
\begin{thm}\label{thm:optimization_formulation_gamma}
For any $\alpha>1$, $\eps\geq 0$, and $\delta\in (0,1)$,  
\begin{align}
    \gamma^\eps_\alpha&(\delta)= \eps + \frac{1}{\alpha-1} \log M(\alpha, \eps, \delta), \label{eq:Optimization_gamma}
\end{align}
	where $\bar p \coloneqq 1-p$ and 
	\begin{equation*}
	    M(\alpha, \eps, \delta)\coloneqq \min_{p\in(\delta, 1)}\left[ p^\alpha (p-\delta)^{1-\alpha}+\bar{p}^\alpha(e^\eps-p+\delta)^{1-\alpha}\right].	
	\end{equation*}
\end{thm}
The proof of this theorem relies on Theorem~\ref{thm:f-divergence_joint-region} and is given in Appendix~\ref{Appendix_gamma_Optimization}.
It can be shown that the term inside the logarithm is convex in $p$ and hence this optimization problem can be numerically solved with an arbitrary accuracy. It seems, however, not simple to analytically derive $\gamma^\eps_\alpha(\delta)$. Nevertheless, we obtain a lower bound in the following theorem that closely approximates $\gamma_\alpha^\eps(\delta)$. We provide its proof in Appendix~\ref{Appendix:Thm_Gamma_LB}. 
\begin{thm}\label{Thm:Lower_BOund_Gamma}
	For any $\eps\geq 0$ and $\alpha>1$, we have 
	\begin{align}
	    \gamma^\eps_\alpha(0)&=0,\\
	    \gamma^\eps_\alpha(\delta) &= \eps-\log(1-\delta), \qquad \qquad~~~~~  \text{if}~~\alpha\delta\geq 1, \label{Eq:Gamma_LB1}\\
	    \gamma^\eps_\alpha(\delta)&\geq \max\{ g(\alpha, \eps, \delta), f(\alpha, \eps, \delta)\}, ~~\textrm{if}~~~ 0<\alpha\delta< 1, \label{Eq:Gamma_LB2}
	\end{align}
where 	
	\begin{equation*}
		\label{eq:Gamma_LB_gfunction}
		g(\alpha, \eps, \delta)\coloneqq \eps - \frac{1}{\alpha-1}\log\frac{\zeta_\alpha}{\delta},
	\end{equation*}
	with $\zeta_\alpha\coloneqq \frac{1}{\alpha}\left(1-\frac{1}{\alpha}\right)^{\alpha-1}$ and 
	\begin{equation*}
		%\label{eq:Gamma_LB_ffunction}
		f(\alpha, \eps, \delta)\coloneqq \eps + \frac{1}{\alpha-1}\log\left(\left(e^{\eps }-\alpha  \delta \right) \left(\frac{\delta -1}{\delta -e^{\eps }}\right)^{\alpha }+\alpha  \delta\right).
	\end{equation*}
\end{thm}

\begin{figure}
  \centering
   \begin{subfigure}{.45\textwidth}
  \centering
  \includegraphics[scale=0.5]{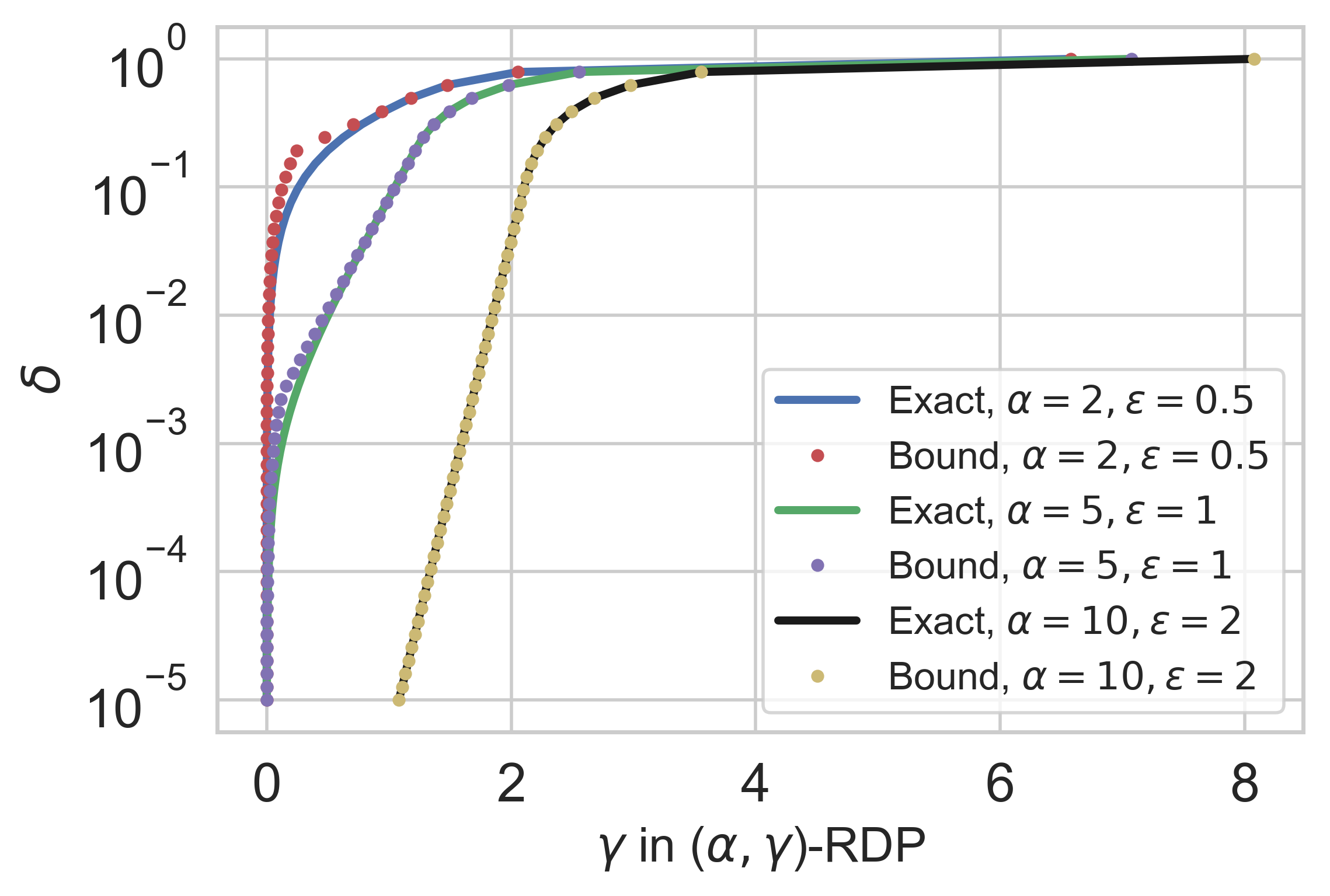}
  \label{fig:Compare_Alpha}
  %\caption{}
  \end{subfigure}
  ~\qquad \begin{subfigure}{.45\textwidth}
  \centering
  \includegraphics[scale=0.5]{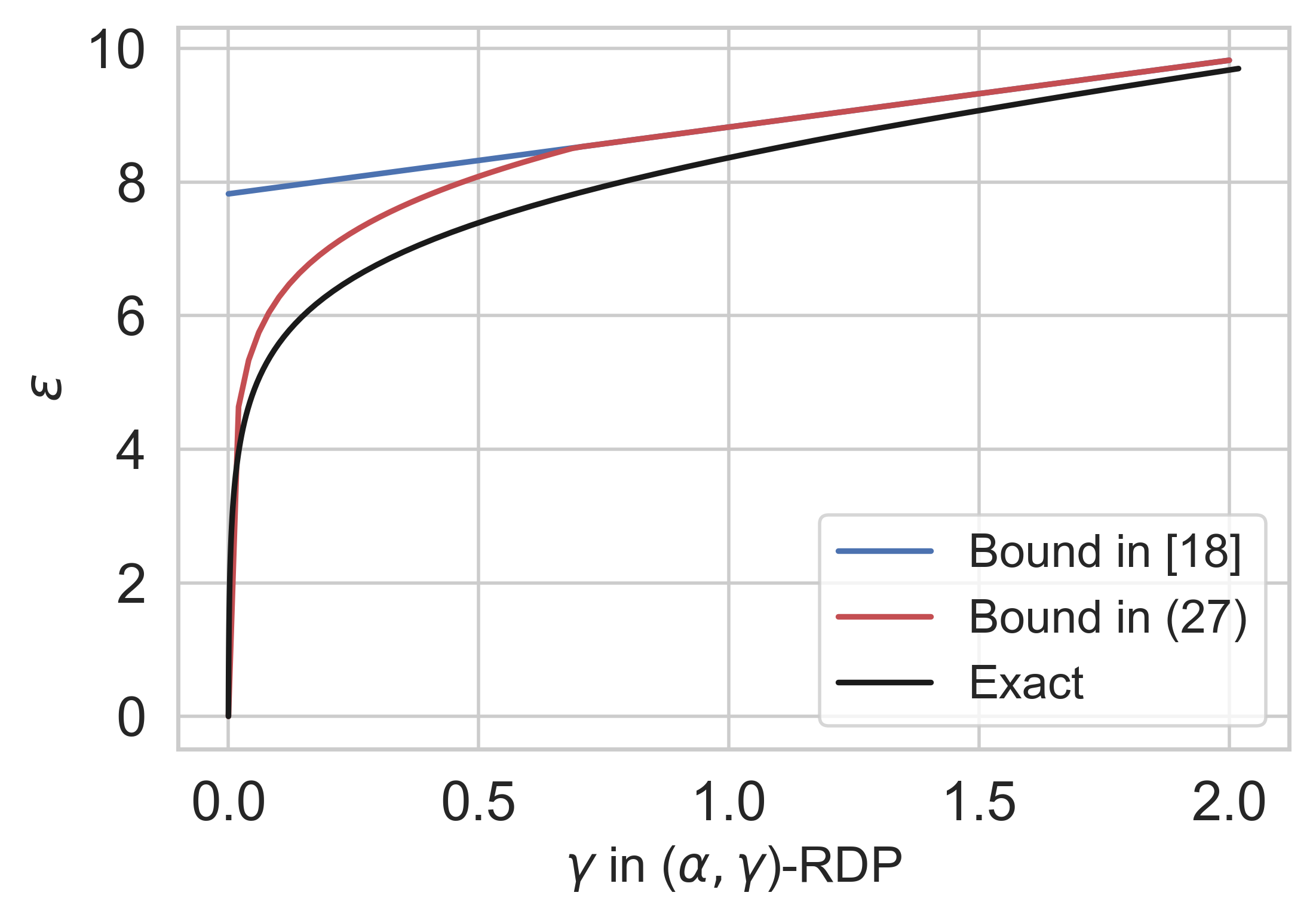}
  \label{fig:Compare_Borja}
  %\caption{}
\end{subfigure}
\caption{Left: The exact values of the map $\delta\mapsto \gamma_\alpha^\eps(\delta)$ obtained via numerically solving convex optimization problem \eqref{eq:Optimization_gamma}. The dotted curves indicate the lower bound on $\gamma_\alpha^\eps(\delta)$ according to Theorem~\ref{Thm:Lower_BOund_Gamma} for three pairs of $(\alpha, \eps)$. Right: Comparison of the exact values of the map $\eps\mapsto \gamma_\alpha^\eps(\delta)$ with the bounds obtained from \eqref{Eq:Approximate_epsilon} and \cite[Theorem 21]{Balle2019HypothesisTI} (i.e., considering only the first term in the minimization in \eqref{Eq:Approximate_epsilon}) with $\alpha =2$ and  $\delta = 10^{-4}$.}
\label{fig:Compare_Alpha}
\end{figure}

In Fig.~\ref{fig:Compare_Alpha} (left panel), we numerically solve \eqref{eq:Optimization_gamma} for three pairs of $(\alpha, \eps)$ and compare them with their corresponding bounds obtained from Theorem~\ref{Thm:Lower_BOund_Gamma}, highlighting the tightness of the above lower bound.
As indicated earlier and illustrated in this figure, the lower bound on $\gamma_\alpha^\eps(\cdot)$ in Theorem~\ref{Thm:Lower_BOund_Gamma} is translated into an upper bound on $\delta_\alpha^\eps(\cdot)$. In practice, it is often more appealing to design differentially private mechanisms with a hard-coded value of $\delta$ (as opposed to the fixed $\eps$). To address this practical need, we convert the lower bound in  Theorem~\ref{Thm:Lower_BOund_Gamma} to an upper bound on $\eps_\alpha^\delta(\cdot)$.
\begin{lem}\label{Lemma_epsilon_Approximate}
	For $\alpha>1$ and  $\gamma\geq 0$, we have 
	\eqn{}{\eps^\delta_\alpha(\gamma)=\left(\gamma +\log(1-\delta)\right)_+, ~~~\textrm{if}~~~\alpha\delta\geq 1,}
	and if $0<\alpha\delta< 1$
	\begin{align}
	    \eps^\delta_\alpha(\gamma)&\leq \frac{1}{\alpha-1}\min\Big\{\Big((\alpha-1)\gamma-\log\frac{\delta}{\zeta_\alpha}\Big)_+, \nonumber\\
	   & \qquad \qquad \qquad \qquad \log\Big(\frac{e^{(\alpha-1)\gamma}-1}{\alpha\delta}+1\Big)
	    \Big\},\label{Eq:Approximate_epsilon}
	\end{align}
	%where $\chi(\gamma)$ is defined in \eqref{tilde_Gamma}. 
	where $\zeta_\alpha$ was defined in Theorem~\ref{Thm:Lower_BOund_Gamma}.
	Moreover, $\eps_\alpha^\delta(0) = 0$.
\end{lem}
This lemma is obtained by solving equality  \eqref{Eq:Gamma_LB1} and inequality \eqref{Eq:Gamma_LB2} for $\eps$. Unlike  $g(\alpha, \eps, \delta)$, the map $\eps\mapsto f(\alpha, \eps, \delta)$ seems complicated to invert. To get around this difficulty, we use the first-order approximation of $f(\alpha, \eps, \delta)$ around $\delta = 0$ to invert the inequality \eqref{Eq:Gamma_LB2}. 
The details are relegated to Appendix~\ref{Appendix:Lemma_epsilon_Approximate}. 
It is worth mentioning that Balle et al. \cite[Theorem 21]{Balle2019HypothesisTI} has recently proved $\eps^\delta_\alpha(\gamma)\leq \gamma - \frac{1}{\alpha-1}\log\frac{\delta}{\zeta_\alpha}$ via a fundamentally different approach. Their bound corresponds to the first term in \eqref{Eq:Approximate_epsilon}, and thus weaker than Lemma~\ref{Lemma_epsilon_Approximate}. To emphasize on the advantage of \eqref{Eq:Approximate_epsilon} over \cite[Theorem 21]{Balle2019HypothesisTI}, we plot these two bounds in Fig.~\ref{fig:Compare_Alpha} (right panel) for $\alpha = 2$ and $\delta = 10^{-4}$. As observed in this figure, considering only the first term in \eqref{Eq:Approximate_epsilon} would lead to non-trivial loss in $\eps$ especially when $\gamma$ is sufficiently small. This observation is analytically justified by the fact that the first term in \eqref{Eq:Approximate_epsilon} does not tend to zero as $\gamma\to 0$ for reasonable values of $\alpha$ and $\delta$ whereas the second term does for any $\delta>0$ and $\alpha>1$.

\begin{remark}\label{remark_Zero_Eps}
As an important special case, this lemma demonstrates that an $(\alpha, \gamma)$-RDP mechanism provides $(0, \delta)$-DP guarantee if $\gamma < \log(\frac{\alpha}{\alpha-1})$ and $\delta\in \big[\zeta_\alpha e^{(\alpha-1)\gamma},\,\frac{1}{\alpha}\big]$.  See Appendix~\ref{Appendix_Remark-eps=0} for the detailed derivation. Notice that this is stronger than what would be obtained from Theorem~\ref{Theorem:RDP_Delta} from which $(0, \delta)$-DP cannot be achieved for  $\gamma>0$. 
\end{remark}

\section{Improved Moments Accountant and Gaussian Mechanisms}\label{Sec:Gaussian}
Moments accountant (MA) was recently proposed by Abadi et al. \cite{Abadi_MomentAccountant} as a method to bypass advanced composition theorems \cite{Vadhan_Composition,Kairouz_Composition}. Given a mechanism $\calM$, the $T$-fold adaptive homogeneous composition $\calM^{(T)}$ is a mechanism that consists of $T$ copies of $\calM$, i.e., $(\calM^1, \dots, \calM^T)$ such that the input of $\calM^i$ may depend on the outputs of $\calM^1, \dots, \calM^{i-1}$.  
Determining the privacy parameters of $\calM^{(T)}$ in terms of those of $\calM$ is an important problem in practice and thus has been subject of an extensive body of research, see e.g., \cite{Vadhan_Composition, Kairouz_Composition, Abadi_MomentAccountant, Balle_AnalyticMomentAccountant}.

Advanced composition theorems \cite{Vadhan_Composition,Kairouz_Composition} are well-known results that provide the DP parameters of $\calM^{(T)}$ for general mechanisms. However, they can be loose and do not take into account the particular noise distribution under consideration (e.g., Gaussian noise). MA was shown to significantly improve upon advanced composition theorems in specific applications such as SGD. The cornerstone of MA is the linear composability of RDP:  If $\calM^1, \dots, \calM^T$ are each $(\alpha, \gamma)$-RDP, then it is shown in  \cite[Theorem 2]{Abadi_MomentAccountant} that $\calM^{(T)}$ is $(\alpha, \gamma T)$-RDP. This result is then translated into DP privacy parameters via Theorem~\ref{Theorem:RDP_Delta}. In general, we assume this holds for \textit{all} $\alpha>1$ and hence   one can obtain the \textit{best} privacy parameters by optimizing over $\alpha$. That is, $\calM^{(T)}$ is $(\eps, \delta)$-DP for any $\eps\geq 0$ and 
\begin{equation}\label{eq:Delta_MA}
    \delta = \inf_{\alpha>1}e^{-(\alpha-1)(\eps- \gamma(\alpha)T)},
\end{equation}
where $\gamma (\alpha) \coloneqq \max_{d\sim d'}D_\alpha(\calM_d\|\calM_{d'})$  is the RDP parameter of the constituent mechanism 
$\calM$ and the dependence on $\alpha$ is made clear. Equivalently, $\calM^{(T)}$ is $(\eps, \delta)$-DP for $\delta\in (0,1)$ and  \begin{equation}\label{eq:eps_MA}
    \eps = \inf_{\alpha>1}\gamma(\alpha)  T - \frac{1}{\alpha-1}\log\delta.
\end{equation}
Since $\alpha\mapsto (\alpha-1)D_{\alpha}(P\|Q)$ is convex \cite[Corollary 2]{van_Erven} for any pair of probability measures $P$ and $Q$, the above minimization is a log-convex problem, and hence, can be solved within an arbitrary accuracy. Furthermore, we show in Section~\ref{Sec:Gaussian} that this minimization has a simple form for Gaussian mechanisms and can be solved analytically.
For the rest of this section, we assume $\calM$ is a Gaussian mechanism and exploit  Lemma~\ref{Lemma_epsilon_Approximate} to derive tighter privacy parameters than \eqref{eq:eps_MA}. 
\begin{figure}
    \centering
     \includegraphics[scale=0.55]{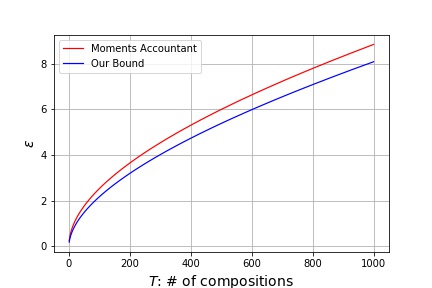}
    \caption{The privacy parameter $\eps$ of the $T$-fold homogeneous composition of Gaussian mechanism each with $\sigma = 20$ according to MA (cf. \eqref{eq:eps_MA2}) and our bound in Lemma~\ref{Lemma:Bound_on_Eps_MA}. We assume $\delta = 10^{-5}$.}
    \label{fig:Composition}
\end{figure}
\subsection{Composition Results for Gaussian Mechanisms}
 Let $f:\mathbbD\to \Reals^n$ be a query function and $\calM$ be a Gaussian mechanism with variance $\sigma^2$; more specifically,  $\mathbbX = \Reals^n$ and $\calM_d = \calN(f(d), \sigma^2\mathrm{I}_n)$ for each $d\in \mathbbD$. For simplicity, we assume that $f$ has unit $L_2$-sensitivity, i.e., $\sup_{d\sim d'}\|f(d)-f(d')\|_2=1$. Since 
 \begin{equation}
     \sup_{d\sim d'}D_\alpha(\calM_d\|\calM_{d'}) = \frac{\alpha}{2\sigma^2}\sup_{d\sim d'}\|f(d)-f(d')\|_2 =\frac{\alpha}{2\sigma^2},
 \end{equation} it follows that $\calM$ is $(\alpha, \gamma(\alpha))$-RDP  for all $\alpha>1$ where $\gamma(\alpha) = \rho\alpha$ and $\rho=\frac{1}{2\sigma^2}$.
In light of the linear composability of RDP, we obtain that  $\calM^{(T)}$, the $T$-fold adaptive composition of $\calM$, is $(\alpha,\gamma(\alpha) T)$-RDP.  Hence, we deduce from \eqref{eq:eps_MA} that $\calM^{(T)}$ is $(\eps, \delta)$-DP for any $\delta\in (0,1)$ and
\begin{equation}\label{eq:eps_MA2}
    \eps=\inf_{\alpha>1}\gamma(\alpha) T - \frac{1}{\alpha-1}\log\delta
     = \rho T + \sqrt{4\rho T\log\frac{1}{\delta}}.
\end{equation}
We next use the machinery developed in the previous section to obtain a tighter bound for the privacy parameter of $\calM^{(T)}$ than \eqref{eq:eps_MA2}. To do so, define 
\begin{equation}\label{def:eps_T}
    \eps^\delta(\rho, T)\coloneqq \inf_{\alpha>1}\eps_\alpha^{\delta}(\rho\alpha T).
\end{equation}
Invoking Lemma~\ref{Lemma_epsilon_Approximate}, we can obtain an upper bound for $\eps^\delta(\rho, T)$.  
\begin{lem}\label{Lemma:Bound_on_Eps_MA}
    The $T$-fold adaptive homogeneous composition of the Gaussian mechanism with variance $\sigma^2$ is $(\eps^\delta(\rho, T), \delta)$-DP with  $\delta\in (0,1)$ and 
    % \eq{\eps^\delta(\rho, T)\leq \inf_{\alpha>1}\zeta_\alphae^{-(\alpha-1)(\eps-\rho\alpha T)},}
    \begin{equation}
	\label{eq:epsilon_gBd}
	\eps^\delta(\rho, T)\leq \min \Big\{\eps_0(\rho, T), ~\eps_1(\rho, T),~ \Big(\frac{\rho T}{\delta} + \log(1-\delta)\Big)_+\Big\},
	\end{equation} 
    where $\rho= \frac{1}{2\sigma^2}$ and 
    \eqn{eps0}{\eps_0(\rho, T)\coloneqq \inf_{\alpha\in(1,\frac{1}{\delta}]}  \left(\rho\alpha T-\frac{1}{\alpha-1}\log\frac{\delta}{\zeta_\alpha}\right)_+,}
    \eqn{eps1}{\eps_1(\rho, T)\coloneqq  \inf_{\alpha\in(1,\frac{1}{\delta}]} \, \frac{1}{\alpha-1}\log\Big(1+\frac{e^{\rho\alpha(\alpha-1)T}-1}{\alpha\delta}\Big),}
    and $\zeta_\alpha$ is as defined in Theorem~\ref{Thm:Lower_BOund_Gamma}.
\end{lem}
% According to this lemma, the $T$-fold composition of the Gaussian mechanism with variance parameter $\sigma^2$ is $(\eps, \delta)$-DP for 
% % \eqn{Eq:Approximate_Delta2}{\delta^\eps_\alpha(\gamma)\leq \frac{e^{(\alpha-1)\rho\alpha}-1}{\alpha\left(e^{(\alpha-1)\eps}-1\right)}.}
% % Hence, we have 
% $$\delta^\eps(\rho, T) \leq  \inf_{\alpha>1}\frac{e^{(\alpha-1)\rho \alpha T}-1}{\alpha\left(e^{(\alpha-1)\eps}-1\right)}.$$

\begin{figure}[t]
	\centering
 	\includegraphics[scale = 0.55]{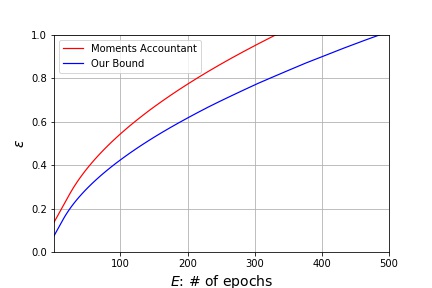}
	\caption{Privacy parameter $\eps$ of noisy SGD algorithm according to MA (cf. \eqref{eq:eps_MA2}) and our bound (cf. 	Lemma~\ref{Lemma:Bound_on_Eps_MA}) for $\delta=10^{-5}$. The parameters of the algorithm are $\sigma = 4$ and the sub-sampling rate $q = 0.001$.}
	\label{fig:SGD}
\end{figure}

The bound given in this lemma can shed light on the optimal variance of the Gaussian mechanism $\calM$ required to ensure that $\calM^{(T)}$ is $(\eps, \delta)$-DP. To put our result in perspective, we first mention two previously-known bounds.
Advanced composition theorems (see, e.g., \cite[Theorem III.3]{Vadhan_Composition}) require $\sigma^2 = \Omega(\frac{T\log(1/\delta)\log(T/\delta)}{\eps^2})$. Abadi et al. \cite[Theorem 1]{Abadi_MomentAccountant} improved this result by showing that $\sigma^2$ suffices to be linear in $T$; more precisely, $\sigma^2 = \Omega(\frac{T\log(1/\delta)}{\eps^2})$. To have a better comparison with our final result, we write this result more explicitly. Plugging $\gamma(\alpha) = \frac{\alpha}{2\sigma^2}$ into \eqref{eq:eps_MA} (or \eqref{eq:Delta_MA}), we can write 
\begin{align}
    \frac{T}{2\sigma^2}&\leq \sup_{\alpha>1}\frac{\eps}{\alpha}+ \frac{1}{\alpha(\alpha-1)}\log\delta\\
    &=\eps - 2\log\delta -2 \sqrt{(\eps-\log\delta)\log\frac{1}{\delta}},
\end{align}
and hence assuming $\delta$ is sufficiently small, we obtain 
\begin{equation}\label{eq:variance_MA}
    \sigma^2\geq \frac{2T}{\eps^2}\log\frac{1}{\delta} + \frac{T}{\eps} + O\left(\frac{1}{\log\delta^{-1}}\right).
\end{equation}
We are now in order to state our result.
\begin{thm}\label{Theorem_noise}
The $T$-fold adaptive homogeneous composition of a Gaussian mechanism with variance $\sigma^2$ is  $(\eps, \delta)$-DP, for $\eps>2\delta\log\frac{1}{\delta}$, if 
\begin{align*}
    \sigma^2&\geq \frac{2T}{\eps^2}\log\frac{1}{\delta} + \frac{T}{\eps} -\frac{2T}{\eps^2}\left(\log(2\log\delta^{-1})+1-\log\eps\right) \\
    & \qquad  + O\left(\frac{\log^2(\log\delta^{-1})}{\log\delta^{-1}}\right).
\end{align*}
\end{thm}
The proof of this theorem is based on a relaxation of Theorem~\ref{Thm:Lower_BOund_Gamma} obtained by ignoring $f(\alpha, \eps, \delta)$. Considering both $f$ and $g$ will result in a stronger result at the expense of more involved analysis. Comparing with \eqref{eq:variance_MA}, Theorem~\ref{Theorem_noise} indicates that, providing $\delta$ is sufficiently small, the variance of each constituent Gaussian mechanism  can  be reduced by $\frac{2T}{\eps^2}\left(\log(2\log\delta^{-1})+1-\log\eps\right)$ compared to what would be obtained from MA.

\subsection{Illustration of Our Bounds}\label{Sec:Abadi_SGD}
 We now empirically compare Lemma~\ref{Lemma:Bound_on_Eps_MA} with the MA guarantee \eqref{eq:eps_MA2} that has been extensively used in the state-of-the-art differentially private algorithms, e.g., \cite{Balle2019mixing, Balle:Subsampling, RDP_SemiSupervised, RDP2_PosteriorSampling, Feldman2018PrivacyAB, Balle_AnalyticMomentAccountant, bhowmick2018protection, McMahan_Privacy}. We do so in two different settings: (1) vanilla $T$-fold composition of the Gaussian mechanisms with fixed variance, and (2) noisy SGD algorithm.
 \begin{description}
     \item[\textbf{Vanilla Gaussian Composition:}] Here, we wish to obtain bounds on the privacy parameter $\eps$ of $\calM^{(T)}$ where $\calM$ is a Gaussian mechanism  with $\sigma=20$. In Fig.~\ref{fig:Composition}, we compare  Lemma~\ref{Lemma:Bound_on_Eps_MA} with MA when $\delta = 10^{-5}$. According to this plot, our result enables us to achieve a smaller privacy parameter by up to $0.75$, i.e., $\max_{T\in [1000]} \eps^\delta_{\mathsf{MA}}(\rho, T)- \eps^\delta(\rho, T) = 0.75$ where $\eps^\delta_{\mathsf{MA}}(\rho, T)$ is the $\eps$ given in \eqref{eq:eps_MA2}. This privacy amplification may have important impacts on recent private deep leaning algorithms. Alternatively, one can observe that our result allows for more iteration for the same $\eps$, e.g., 100 more iterations for any $\eps$ larger than $6$.
     \item[Noisy SGD:] SGD is the standard algorithm for training many machine learning models. In order to fit a model without compromising privacy, a standard practice is to add Gaussian noise to the gradient of each mini-batch, see e.g., \cite{Abadi_MomentAccountant, Shokri:DeepLearning, Feldman2018PrivacyAB, Balle2019mixing,chaudhuri2011differentially, Bassily1, BAssily2}. The prime use of  MA  was to exploit the RDP's simple composition property in deriving the privacy parameters of the  noisy SGD algorithm \cite[Algorithm 1]{Abadi_MomentAccountant}. To have a fair comparison, we analyze this algorithm (see Algorithm~\ref{alg:noisySGD}) with the sub-sampling rate $q=0.001$ and noise parameter $\sigma = 4$ and then compute its DP parameter via \eqref{eq:eps_MA2} with $\rho=\frac{q^2}{(1-q)\sigma^2}$ (see \cite[Lemma 3]{Abadi_MomentAccountant}) and $\delta = 10^{-5}$. We then compare it in Fig.\ \ref{fig:SGD} with Lemma~\ref{Lemma:Bound_on_Eps_MA} with the same $\rho$ and $\sigma$. As demonstrated in this figure, our result allows remarkably more epochs (often over a hundred) within the same privacy budget and thus providing higher utility. 
 \end{description}
\begin{algorithm}
\caption{Noisy SGD}
\label{alg:noisySGD}
\begin{algorithmic}[1]
%\REQUIRE{}
\STATE{{\bf Input:} Dataset $d=\{x_1, \dots, x_n\}$, loss function $\ell(\theta, x)$, initial point $\theta_0$, batch size $m$, noise variance $\sigma^2$, and clipping threshold $C$.}
\FOR{$t= 1, \dots, T$}
\STATE{Select randomly a batch $\mathsf I_t\subset [n]$ of size $m$}
\STATE{$g_t(x_i) = \nabla_\theta\ell(\theta_t, x_i),$ for $i\in \mathsf I_t$}
\STATE{$\tilde g_t(x_i) = g_t(x_i)\min\{1, \frac{C}{\|g_t(x_i)\|_2}\}$}
\STATE{$\theta_{t+1} = \theta_t - \frac{\eta_t}{m}\left[\sum_{i\in \mathsf I_t}\tilde g_t(x_i) + \sigma C Z\right]$, ~~ $Z\sim \calN(0, \mathrm{I})$}   
\ENDFOR
\STATE{{\bf Output} $\theta_T$}
\end{algorithmic}
\end{algorithm}

Since Lemma~\ref{Lemma_epsilon_Approximate} is shown to improve on the composition results of MA, it is reasonable to  construct the \textit{improved} MA: First use the linear composability of RDP to take into account the composition and then use Lemma~\ref{Lemma_epsilon_Approximate} to convert the resulting RDP guarantee to $(\eps, \delta)$-DP. We next show that improved MA might lead to tighter guarantee than hypothesis test privacy.

\subsection{Comparison with $f$-DP} 
As mentioned earlier, $f$-DP (cf. Definition~\ref{Def:f_DP}) leads to stronger DP guarantee than what is obtained by MA  for noisy SGD algorithms. 
More precisely, Bu et al. \cite[Theorem 2]{GaussianDP_deep} showed that if one applies composition results of $f$-DP (i.e., \cite[Theorem 3.2]{GaussianDP}) to noisy SGD algorithms and then converts it to DP (via \cite[Proposition 3.12]{GaussianDP}), then the resulting $\eps$ is asymptotically smaller than \eqref{eq:eps_MA2} for  any $\delta\in (0,1)$ provided that the sub-sampling rate $q$ is scaled as $\frac{1}{\sqrt{T}}$ with $T$ being the number of iteration. A natural question raised here is whether this result still holds if we replace MA with the improved MA.

In Fig.~\ref{fig:MA_GDP}, we consider noisy SGD algorithm with Gaussian noise with $\sigma = 0.6$ and sub-sampling rate $q = 0.003$ (similar to  \cite[Fig. 2]{GaussianDP_deep}) and compare Lemma~\ref{Lemma:Bound_on_Eps_MA} with \cite[Theorem 2]{GaussianDP_deep}. As clearly illustrated by this figure, the improved MA may yield tighter privacy guarantees than what $f$-DP promises.
\begin{figure}[h]
	\centering
\includegraphics[scale = 0.5]{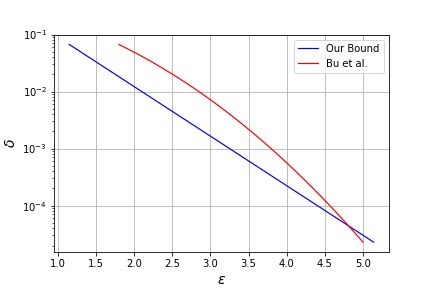}
	\caption{Comparison of parameters of $\eps$ and $\delta$ in noisy SGD algorithm obtained from Lemma~\ref{Lemma:Bound_on_Eps_MA} and  \cite[Theorem 2]{GaussianDP_deep}. The parameters of the algorithm are as follows: $q = 0.003$, epoch $E = 30$ (hence  $T = \frac{E}{q} = 10000$), and $\sigma = 0.6$. }
	\label{fig:MA_GDP}
\end{figure}

\section{Hypothesis Testing Privacy}\label{Section:BHT}
In this section, we investigate the relationship between RDP and hypothesis test privacy, that is, we focus on Question Two in the introduction.  Let $X$ be the output of a mechanism $\calM$. For any pair of neighboring dataset $d\sim d'$, we consider the hypothesis test (repeated from the introduction for convenience)
\begin{align}\label{Hypothesis2}
    H_0:~ X \sim \calM_d~~~\text{vs.}~~~
    H_1:~ X \sim \calM_{d'}.
\end{align}
The fundamental efficiency of a randomized test between $H_0$ and $H_1$ is delineated by a \textit{decision rule}, a random transformation $P_{Z|X}:\mathbbX \to \mathcal P(\{0,1\})$ where $1$ indicates that $H_0$ is rejected. Type I and type II error probabilities corresponding to the decision rule $P_{Z|X}$ are given by $\int P_{Z|X}(1|x)\calM_d(\text{d}x)$ and $\int P_{Z|X}(0|x)\calM_{d'}(\text{d}x)$, respectively. To capture the optimal tradeoff between type I and type II error probabilities, it is customary to define \textit{tradeoff function} $\beta^{dd'}_\calM:[0,1]\to [0,1]$ given by
\begin{equation}\label{Def:beta_BHT2}
    \beta^{dd'}_\calM(\tau) \coloneqq  \inf~\int P_{Z|X}(0|x)\calM_{d'}(\text{d}x)
    %\inf\left\{\int P_{Z|X}(0|x)\calM_{d'}(\text{d}x)~:~\int P_{Z|X}(1|x)\calM_d(\text{d}x)\leq \tau \right\},
\end{equation}
where the infimum is taken over all decision rules $P_{Z|X}$ such that $\int P_{Z|X}(1|x)\calM_d(\text{d}x)\leq \tau$.
% \begin{figure}[h]
%     \centering
%     \includegraphics[width=8.7cm, height=1.3cm]{BHT.png}
%     \caption{The schematic view of binary hypothesis test \eqref{Hypothesis2} where $d'\in \mathbb D$ is an arbitrary dataset neighboring to $d$ . The decision rule $P_{Z|X}$ maps $X$ to a binary random variable $Z$ that (probabilistically) specifies the hypothesis: $Z=0$ corresponds to $H_0$ and $Z=1$ to $H_1$. \SA{I don't think this fig is need. It's a basic hypothesis testing scenario. But if you think it'd make the discussion clear, I keep it.}}
%     \label{fig:BHT}
% \end{figure}

Note that we can always assume, without loss of generality, that $\tau + \beta^{dd'}_\calM(\tau)\leq 1$, since for any decision rule one can take its negation. The line $\tau + \beta^{dd'}_\calM(\tau) = 1$ indicates the complete indistinguishability between $d$ and $d'$ on the basis of a mechanism's output. It follows from the definition that the map $\tau \mapsto \beta_\calM^{dd'}(\tau)$ is non-increasing and convex. Recall that the mechanism $\calM$ is said to be $f$-DP for a convex and non-increasing function $f$ that is is majorized by $\inf_{d\sim d'}\beta^{dd'}_\calM$, that is if $f(\tau)\leq \inf_{d\sim d'}\beta^{dd'}_\calM(\tau)$ for any $\tau\in [0,1]$. Hence, the problem of determining the relationship between RDP and $f$-DP reduces to characterizing the set  $(\tau, \beta)$ such that $\beta\geq \inf_{d\sim d'}\beta_\calM^{dd'}(\tau)$ for all mechanisms $\calM$ with a certain level of RDP guarantee.  To this goal, we  define the \textit{privacy region} of mechanism $\calM$ as
$$\calC_\calM\coloneqq \bigcup_{d\sim d'}\{(\tau, \beta)\in [0,1]^2:~\beta_\calM^{dd'}(\tau)\leq\beta\leq 1-\tau\}.$$
It was shown by \cite{Wasserman, Kairouz_Composition} that a mechanism $\calM$ is $(\eps, \delta)$-DP if and only if 
\begin{equation}\label{def:C_eps_delta}
\calC_\calM\subseteq
    \calC(\eps, \delta)\coloneqq \{(\tau, \beta)\in [0, 1]^2: \tau + e^\eps \beta\geq 1-\delta, \beta + e^\eps \tau\geq \bar\delta\}.
\end{equation}
\begin{figure}[t]
	\centering
 	\includegraphics[scale = 0.45]{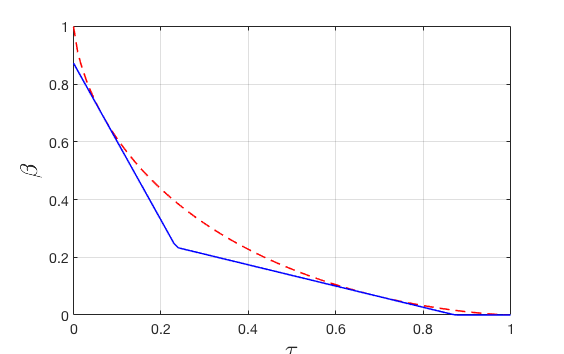}
	\caption{Two outer bounds for the Gaussian mechanism with $\sigma^2=1$: The red curve is the map $\tau \mapsto \Phi\left(\Phi^{-1}(\bar \tau)-1/\sigma\right)$ and the blue curve specifies the region $\calC(\eps, \delta_\eps)$ for $\eps=1$ and $\delta_\eps$ given in \eqref{Eq:Delta_gaussian}.}
	\label{fig:Gaussian}
\end{figure}

\begin{remark}\label{Remark:duality}
Recall from the definition of $\sE_\lambda$-divergence \eqref{Defi_HS_Divergence} that, for any pair of distributions $(P,Q)$ and positive $\lambda$, we have $\sE_\lambda(P\|Q) = P(\frac{\text{d}P}{\text{d}Q}\geq \lambda)-\lambda Q(\frac{\text{d}P}{\text{d}Q}\geq \lambda)$.
Since according to Neyman-Pearson lemma $\beta_\calM^{dd'}(\tau) = \calM_{d'}(\log\frac{\text{d}\calM_d}{\text{d}\calM_{d'}}\geq \eps)$ where $\tau = \calM_{d}(\log\frac{\text{d}\calM_d}{\text{d}\calM_{d'}}\leq\eps)$, it follows that the line $\beta = e^{-\eps}(1-\tau - \delta_\eps)$, with $\delta_\eps \coloneqq  \min_{d\sim d'}\sE_{e^\eps}(\calM_d\|\calM_{d'})$, supports $\calM$ from below. Swapping $d$ and $d'$, we deduce that the line $\beta = 1-\delta_\eps - e^{\eps}\tau $ is another supporting line of $\calC_\calM$ with slope $e^{\eps}$. Due to the convexity of $\calC_\calM$, the collection of all supporting lines losslessly constructs $\calC_\calM$; thus,  $\bigcap_{\eps\geq 0}\calC(\eps, \delta_\eps)  = \calC_\calM$.  In other words, the collection of $\{(\eps, \delta_\eps)\}_{\eps\geq 0}$ and the mapping $\tau\mapsto \inf_{d\sim d'}\beta^{dd'}_\calM(\tau)$ capture the same privacy guarantee. This provides a new lens to explore, delineate and interpret privacy guarantee achieved by differential privacy. This new perspective has recently been adopted by Dong et al. \cite{GaussianDP}.
To illustrate this observation, consider the Gaussian mechanism. It is easy to see that for Gaussian mechanisms (assuming unit $L_2$-sensitivity)
\begin{equation}\label{Eq:Delta_gaussian}
    \delta_\eps=  \Phi\Big(-\eps\sigma+\frac{1}{2\sigma}\Big)-e^\eps\Phi\Big(-\eps\sigma-\frac{1}{2\sigma}\Big),
\end{equation} 
where $\Phi$ is the standard normal CDF. On the other hand, for a Gaussian mechanism $\calM$ with variance $\sigma^2$, the Neyman-Pearson lemma implies that the tradeoff function $\inf_{d\sim d'}\beta^{dd'}_\calM(\tau) = G_{\frac{1}{\sigma}}(\tau)$, where   \begin{equation}\label{Eq:Gaussian_beta}
     G_\mu(\tau) = \Phi\Big(\Phi^{-1}(1-\tau)-\mu\Big),
\end{equation}
and $\Phi^{-1}$ is the inverse of $\Phi$. It is worth mentioning that $G_\mu(\tau)$ in fact corresponds to the smallest type II error probability of testing $\calN(0, 1)$ against $\calN(\mu, 1)$ with type I error probability being $\tau$. 
In Fig.~\ref{fig:Gaussian}, we identify the region $\calC_\calM$ by its lower boundary (red curve) given by the above tradeoff function and its upper boundary $\beta = 1-\tau$. The blue curve is the lower boundary of $\calC(\eps, \delta_\eps)$ for $\eps = 1$. 
\end{remark}

While the DP constraint can be operationally interpreted via \eqref{def:C_eps_delta}, it is not clear how to obtain a similar interpretation for RDP constraint. %On a rather negative note, it has been recently shown in \cite{Balle2019HypothesisTI} that $(\alpha, \gamma)$-RDP constraint cannot be equivalently expressed in terms of hypothesis testing problem \eqref{Hypothesis2}.  
Nevertheless, we wish to obtain some implications of a mechanism's RDP constraints on its privacy regions. We begin by giving an explicit formula for the RDP guarantee of a mechanism in terms of the derivative of the map $\tau\mapsto \beta_\calM^{dd'}(\tau)$ for $d\sim d'$. 
\begin{prop}\label{Prop:Polyankiy_RDP}
Given $\alpha>1$, a mechanism $\calM$ is $(\alpha, \gamma)$-RDP for 
\begin{equation}
    \gamma = \sup_{d\sim d'}\frac{1}{\alpha-1}\log\left(1-\beta_\calM^{dd'}(0) + \int_0^1|\Gamma_{dd'}(\tau)|^{1-\alpha}\textnormal{d}\tau\right),
\end{equation}
where $\Gamma_{dd'}(\tau) = \frac{\textnormal{d}}{\textnormal{d}\tau}\beta^{dd'}_\calM(\tau)$.
\end{prop}
The proof of this result relies on a general fact:  all $f$-divergences between $\calM_d$ and $\calM_{d'}$ can be explicitly expressed in terms of the derivative of $\beta^{dd'}_{\calM}$. This was mentioned, without a proof, in \cite[Eq.  (2.79)]{Polyanskiy_thesis} in a completely different context and was recently proved in \cite[Proposition B.4]{GaussianDP}. 
We give a more direct proof in Appendix~\ref{Appendix:Polyankiy_RDP}.

Proposition~\ref{Prop:Polyankiy_RDP} provides an explicit RDP guarantee for a mechanism with a given hypothesis test privacy constraint. The other direction seems more practical: Given an $(\alpha, \gamma)$-RDP mechanism, what can we say about its privacy region $\calC_\calM$? 
There are two approaches to address this question. First, one can use the machinery developed in Section~\ref{Sec:Conversion} to relate $(\alpha, \gamma)$-RDP constraint to $(\eps, \delta^\eps_\alpha(\gamma))$-DP and then declare $\calC(\eps, \delta^\eps_\alpha(\gamma))$
as an outer bound for the privacy region for any $\eps\geq 0$. Alternatively, one can use information theoretic results (such as data processing inequality) to \textit{directly} relate R\'enyi divergence to  type I and type II error probabilities in hypothesis testing \eqref{Hypothesis2} (see, e.g., \cite{Arimoto_Polyankiy}). In the following, we delineate these two approaches.

Since all $(\alpha, \gamma)$-RDP mechanisms are $(\eps, \delta^\eps_\alpha(\gamma))$-DP, we immediately obtain the following result from \eqref{def:C_eps_delta}.   

\begin{lem}\label{lemma:region_C_Eps}
Let $\calM$ be an $(\alpha, \gamma)$-RDP mechanism.  Then, we have
\begin{equation}
    \calC_\calM \subseteq \bigcap_{\eps\geq 0}\calC(\eps, \delta^\eps_\alpha(\gamma)).
\end{equation} 
\end{lem}
% \begin{proof}
% Recall that all mechanisms in $\mathbb M_\alpha(\gamma)$ are $(\eps_\alpha^\delta(\gamma), \delta)$-DP for $\delta\in (0,1)$. Invoking the fact that the privacy regions of $(\eps, \delta)$-DP mechanisms are contained in $\calC(\eps, \delta)$, the result follows. 
% \end{proof}
Note that since $\eps\mapsto \delta^\eps_\alpha(\gamma)$ characterizes the DP parameters of the worst mechanism in $\mathbb M_\alpha(\gamma)$, it follows that the privacy regions of \textit{all} $(\alpha, \gamma)$-RDP mechanisms are contained in $\bigcap_{\eps\geq 0} \calC(\eps, \delta^\eps_\alpha(\gamma))$, or equivalently, $\bigcup_{\calM\in \mathbb M_\alpha(\gamma)} \calC_\calM \subseteq  \bigcap_{\eps\geq 0} \calC(\eps, \delta^\eps_\alpha(\gamma)).$

Instead of dealing with the infinite collection of $(\eps, \delta_\alpha^\eps(\gamma))$
and taking the intersection of $\calC(\eps, \delta_\alpha^\eps(\gamma))$, we can alternatively focus on the tradeoff function (cf.\ Remark~\ref{Remark:duality}). 
% A mechanism $\calM$ is said to be $\mu$-GDP if distinguishing $\calM_d$ from $\calM_{d'}$ given an observation is at least as hard as distinguishing $\calN(0, 1)$ from $\calN(\mu, 1)$. Since the type I-type II tradeoff function for the latter hypothesises testing problem is given by $\Phi(\Phi^{-1}(1-\tau)-\mu)$, a mechanism $\calM$ is $\mu$-GDP if 
% \begin{equation}\label{Def:GDP_G}
%   \min_{d\sim d'}\beta_\calM^{dd'}(\tau)\geq \Phi(\Phi^{-1}(1-\tau)-\mu),  
% \end{equation}
% for all $\tau\in [0,1]$.
That is, we wish to study the privacy regions of RDP mechanisms by \textit{directly} computing bounds on the tradeoff function rather than converting RDP into $(\eps, \delta)$-DP. Adopting this viewpoint,  we establish two outer bounds for the privacy region of an $(\alpha, \gamma)$-RDP mechanism in the following lemma.
\begin{figure}[t]
		\begin{subfigure}[t]{0.3\textwidth}
		\centering
		\includegraphics[scale = 0.16]{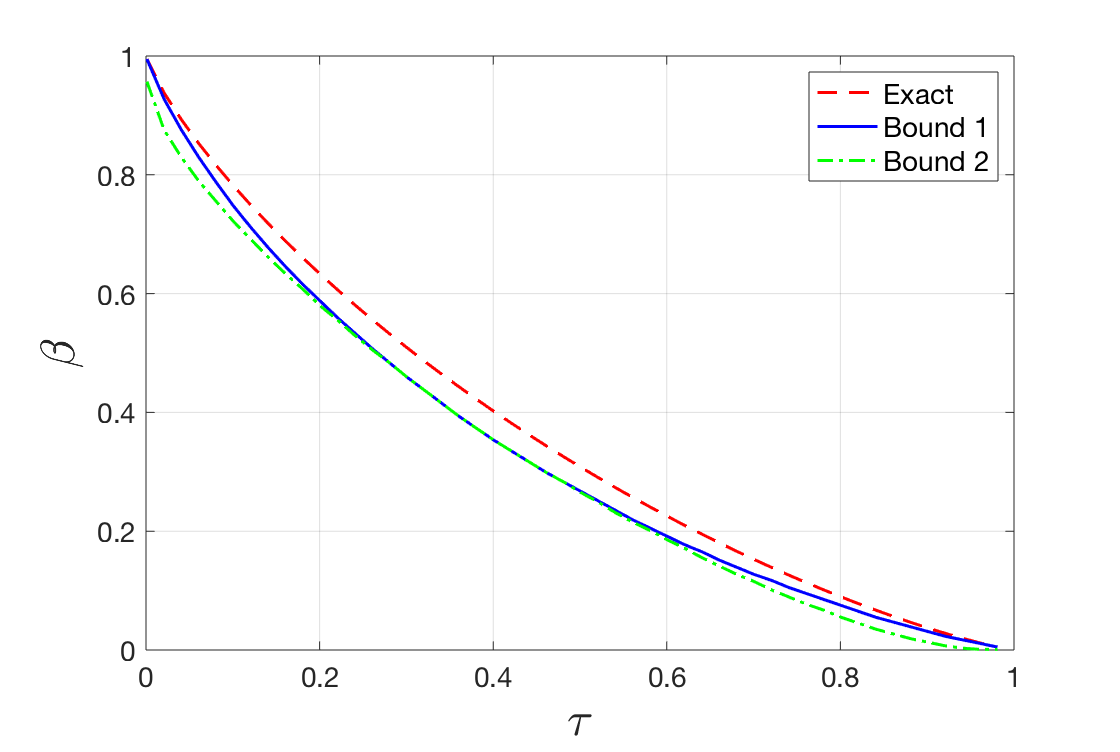}
		\caption{$\sigma = 2, T=1$}
		\end{subfigure}
		\qquad
		\begin{subfigure}[t]{0.3\textwidth}
		\centering
		\includegraphics[scale = 0.16]{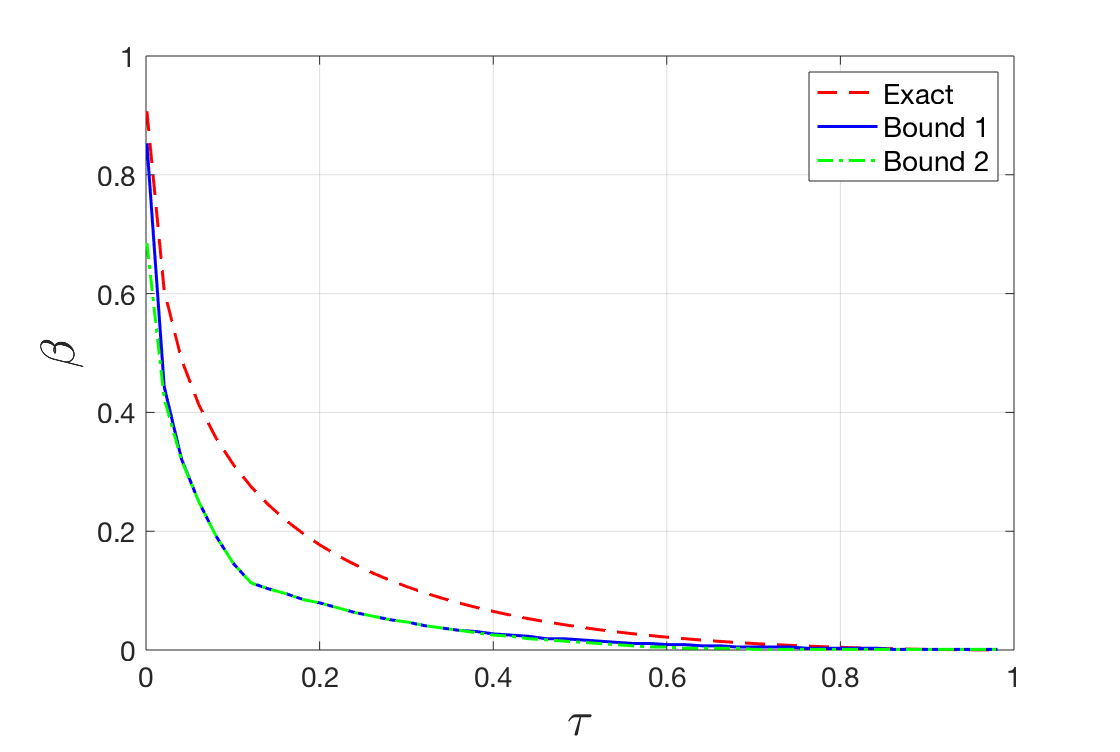}
		\caption{$\sigma = 4, T=50$}
		\end{subfigure}
		\qquad 
        \begin{subfigure}[t]{0.3\textwidth}
		\centering
		\includegraphics[scale = 0.16]{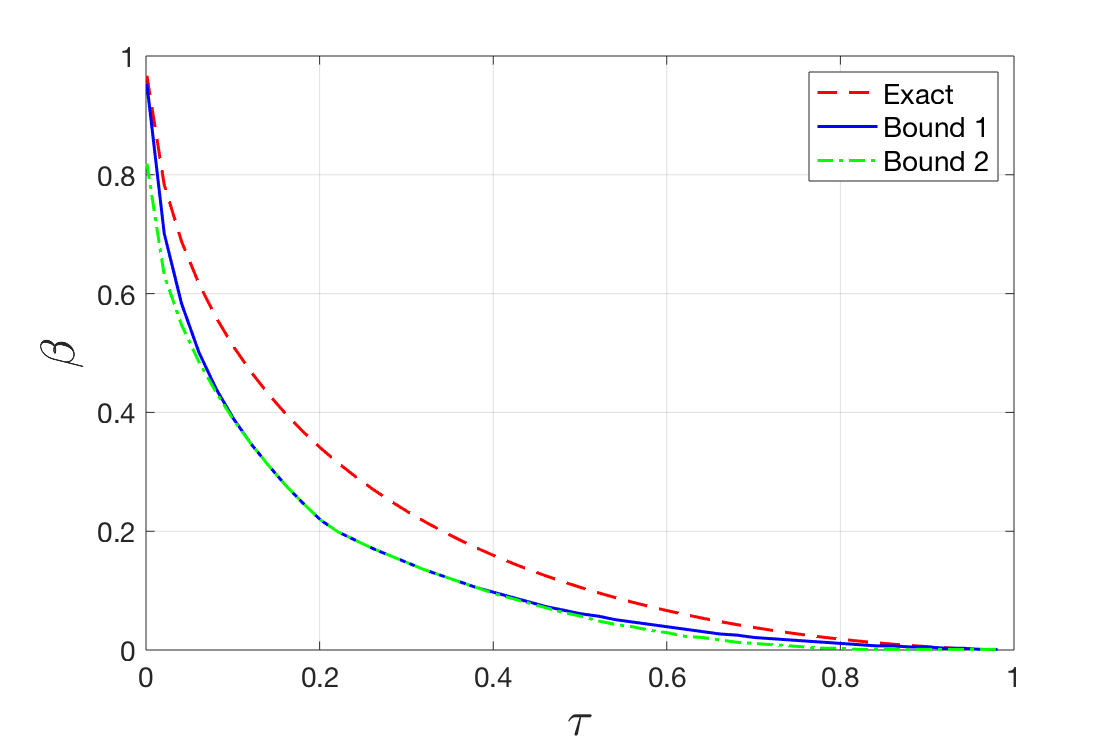}
		\caption{$\sigma = 8, T=100$}
		\end{subfigure}
	\caption{The outer bounds for the privacy region of the $T$-fold homogeneous Gaussian mechanism. The regions marked as  Bound 1 and Bound 2 correspond to \eqref{Bound1_Corollary} and \eqref{Bound2_Corollary}, respectively and the region marked as Exact corresponds to \eqref{eq:PriRegion_GDP}. Recall that the privacy ``regions'' are to be interpreted as the region between depicted curves and the diagonal line $\tau + \beta =1$.}
	\label{fig:RDP_Region2}
\end{figure}
\begin{lem}\label{Lemma:RDP_Region_DAlpha}
Let $\calM$ be an $(\alpha, \gamma)$-RDP mechanism. Then, the privacy region of $\calM$ satisfies  
\begin{align}
    \calC_\calM&\subseteq \{(\tau, \beta)\in (0, 1)^2: d_\alpha(\bar\tau\|\beta)\leq \gamma, d_\alpha(\bar\beta\|\tau)\leq \gamma\}\label{Subset1}\\
   &\subseteq \{(\tau, \beta)\in (0, 1)^2: d(\bar \tau\|\beta)\leq \gamma, d(\bar\beta\|\tau)\leq \gamma\},\label{Subset2}
\end{align}
% we have for all $d\sim d'$ 
% $$\max\{d_\alpha(1-\tau\|\beta^{dd'}_\calM(\tau)), d_\alpha(\beta^{dd'}_\calM(\tau)\|1-\tau)\}\leq \gamma,$$
where $d(a\|b) = a\log\frac{a}{b}+\bar a\log\frac{\bar a}{\bar b}$ and $d_\alpha(a\|b)\coloneqq \frac{1}{\alpha-1}\log\left(a^\alpha b^{1-\alpha}+\bar a^\alpha \bar b^{1-\alpha}\right)$ for $a, b\in (0,1)$.
\end{lem}
\begin{proof}
Let $P_{Z|X}$ be an optimal randomized test mapping the mechanism's output $X$ to a binary variable $Z$ corresponding to $H_0$ and $H_1$, i.e., $\int P_{Z|X}(0|x)\calM_d(\text{d}x) = 1-\tau$ and $\int P_{Z|X}(0|x)\calM_{d'}(\text{d}x) = \beta^{dd'}_\calM(\tau)$. (The existence of such an optimal randomized test is guaranteed by Neyman-Pearson lemma.) Due to the data processing inequality, we have 
\begin{align}
    D_\alpha(\calM_d\|\calM_{d'})&\geq D_\alpha(\mathsf{Bernoulli}(\tau)\|\mathsf{Bernoulli}(1-\beta^{dd'}_\calM(\tau)))\nonumber\\
    & = d_\alpha(1-\tau\|\beta^{dd'}_\calM(\tau)),\label{Eq:DPI}
\end{align}
This in turn implies that  for all $d\sim d'$ 
\begin{equation}
    \max\{d_\alpha(1-\tau\|\beta^{dd'}_\calM(\tau)), d_\alpha(\beta^{dd'}_\calM(\tau)\|1-\tau)\}\leq \gamma,
\end{equation}
which in turn implies \eqref{Subset1} by noticing that $a\mapsto d_\alpha(a\|b)$ is decreasing for $a<b$ and similarly $b\mapsto d_\alpha(a\|b)$ is decreasing for $b<a$. Since $\alpha\mapsto D_\alpha(P\|Q)$ is non-decreasing \cite[Theorem 3]{van_Erven}, the inclusion \eqref{Subset2} follows  immediately.
\end{proof}
It is worth mentioning that $d_\alpha(a\|b)$ is  closely related to \cite[Definition 9]{Balle2019HypothesisTI}.  
%The region characterized in Lemma~\ref{Lemma:RDP_Region_DAlpha} is depicted in Fig.~\ref{fig:RDP_Region1} for different $\alpha$'s and $\gamma$'s. 
% \begin{figure}[t]
% 	\centering
%  	\includegraphics[height = 0.3\linewidth, width = 0.4\linewidth]{RDP-Region1.png}\qquad \qquad
%  	\includegraphics[height = 0.3\linewidth, width = 0.4\linewidth]{RDP-Region2.png}
% 	\caption{The outer bounds for privacy regions of $(\alpha, \gamma)$-RDP mechanisms described in Lemma~\ref{Lemma:RDP_Region_DAlpha}.}
% 	\label{fig:RDP_Region1}
% \end{figure}
Note that although the set in \eqref{Subset2} strictly contains the one in \eqref{Subset1}, it enables us to derive a simple outer bound for the privacy region of mechanisms when optimizing over $\alpha$. This is formalized in the following result which is an immediate corollary of Lemma~\ref{Lemma:RDP_Region_DAlpha}. 
\begin{cor}\label{Corollay_Gaussian_Region}
If mechanism $\calM$ is $(\alpha, \gamma(\alpha))$-RDP for all $\alpha>1$. Then its privacy region satisfies  
%The privacy region of Gaussian mechanism with variance $\sigma^2$ satisfies
\begin{align}
\calC_\calM &\subseteq\bigcap_{\alpha>1}\{(\tau, \beta): d_\alpha(\bar\tau\|\beta)\leq \gamma(\alpha),d_\alpha(\bar\beta\|\tau)\leq \gamma(\alpha)\}\label{Bound11_Corollary}\\
&\subset\bigcap_{\alpha>1}\{(\tau, \beta): d(\bar\tau\|\beta)\leq \gamma(\alpha),d(\bar\beta\|\tau)\leq \gamma(\alpha) \}.\label{Bound22_Corollary}
%&\subseteq \left\{(a, b)\in [0, 1]^2: d(a+b\|1-a)-\nu d(a+b\|b)\leq \rho, d(a+b\|1-a)-\frac{1}{\nu} d(a+b\|b)\geq \frac{1}{\nu}\rho \right\},
\end{align}
\end{cor}
To demonstrate the accuracy of Corollary~\ref{Corollay_Gaussian_Region}, we consider Gaussian mechanisms for the remainder of this section. Recall that the Gaussian mechanism with variance $\sigma^2$ is $(\alpha, \gamma)$-RDP for $\gamma = \rho \alpha$ with $\rho = \frac{1}{2\sigma^2}$. Recall that the $T$-fold composition of such mechanism is $(\alpha, \rho \alpha T)$-RDP, implying that $\calM^{(T)}$ is a Gaussian mechanism with variance $\frac{\sigma^2}{T}$. Hence, according to \eqref{Eq:Gaussian_beta}, we have 
\begin{equation}\label{eq:Def_G_rho}
    \inf_{d\sim d'}\beta^{dd'}_{\calM^{(T)}}(\tau) = G_{\sqrt{2\rho T}}(\tau) %\Phi\left(\Phi^{-1}(1-\tau)-\sqrt{2\rho T}\right),
\end{equation}
This, in turn, implies that $\calC_{\calM^{(T)}}$ the privacy region of $\calM^{(T)}$ is given by 
\begin{equation}\label{eq:PriRegion_GDP}
\calC_{\calM^{(T)}} =\left\{(\tau, \beta)\in (0, 1)^2: G_{\sqrt{2\rho T}}(\tau)\leq \beta\leq 1-\tau\right\}.
\end{equation}
Specializing Corollary~\ref{Corollay_Gaussian_Region} to $\calM^{(T)}$, we can express outer bounds given in \eqref{Bound11_Corollary} and \eqref{Bound22_Corollary} as  
 \begin{align}
     \calC_{\calM^{(T)}} &\subseteq\bigcap_{\alpha>1}\{(\tau, \beta): d_\alpha(\bar\tau\|\beta)\leq \rho\alpha T, d_\alpha(\bar\beta\|\tau)\leq \rho\alpha T \}\label{Bound1_Corollary}\\
&\subset\left\{(\tau, \beta)\in [0, 1]^2: d(\bar\tau\|\beta)\leq \rho T, d(\bar\beta\|\tau)\leq \rho T \right\}.\label{Bound2_Corollary}
 \end{align}

In Fig.~\ref{fig:RDP_Region2}, we compare these outer bounds with the exact privacy region given in \eqref{eq:PriRegion_GDP}. Note that the region \eqref{Bound22_Corollary}, while being weaker than the region in \eqref{Bound11_Corollary}, can be explicitly characterized for Gaussian mechanisms. 
\begin{figure}[t]
	\centering
	\begin{subfigure}[t]{0.3\textwidth}
		\centering
		\includegraphics[scale = 0.16]{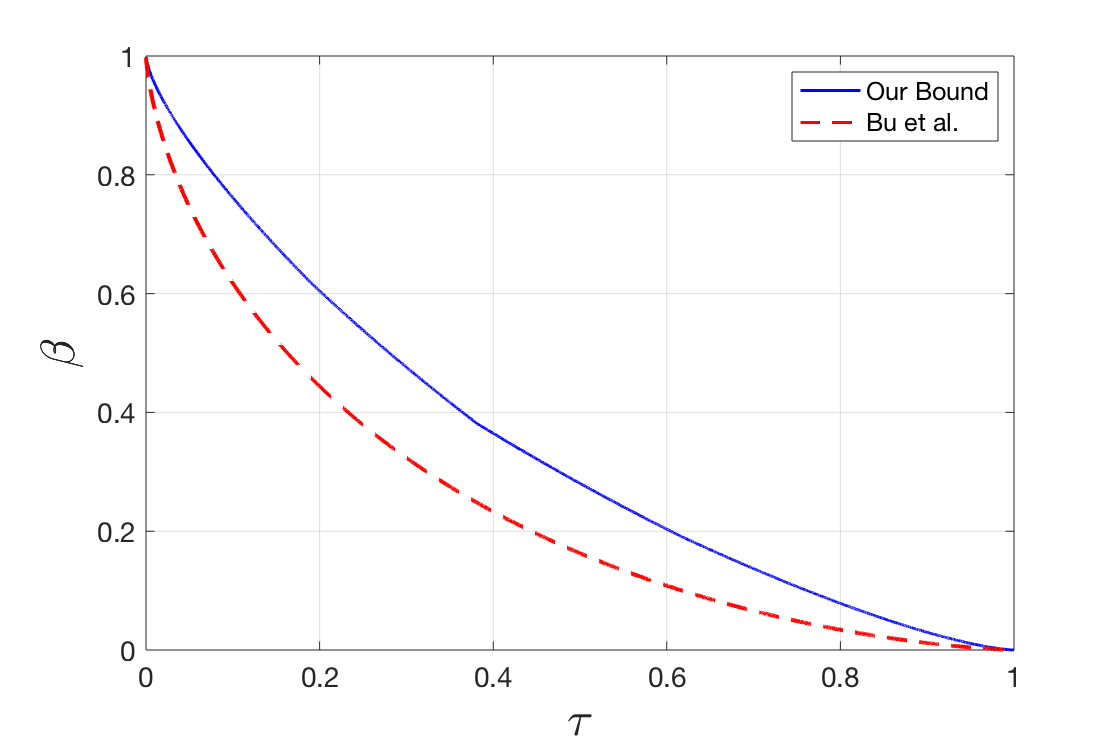}
		\caption{$\sigma = 0.6$, $Tq=15$}
% 		\label{fig:T=45}
	\end{subfigure}
	\qquad 
	\begin{subfigure}[t]{0.3\textwidth}
		\centering
		\includegraphics[scale = 0.16]{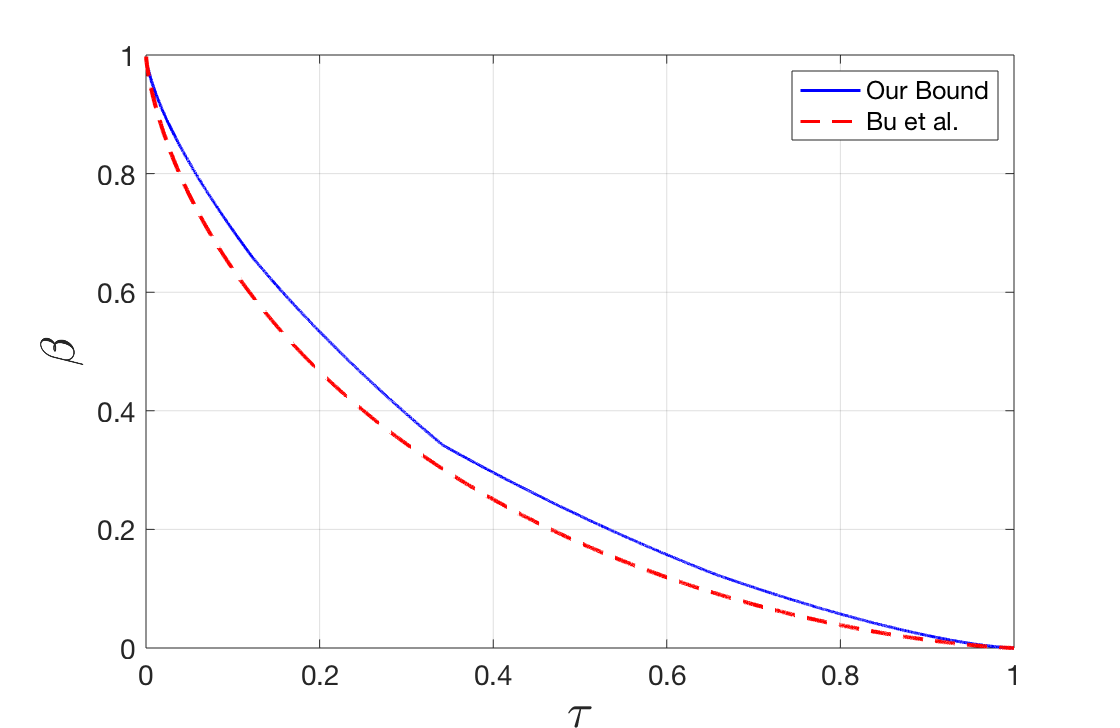}
		\caption{$\sigma = 0.7$, $Tq=30$}
	\end{subfigure}
\qquad
	\begin{subfigure}[t]{0.3\textwidth}
		\centering
		\includegraphics[scale = 0.16]{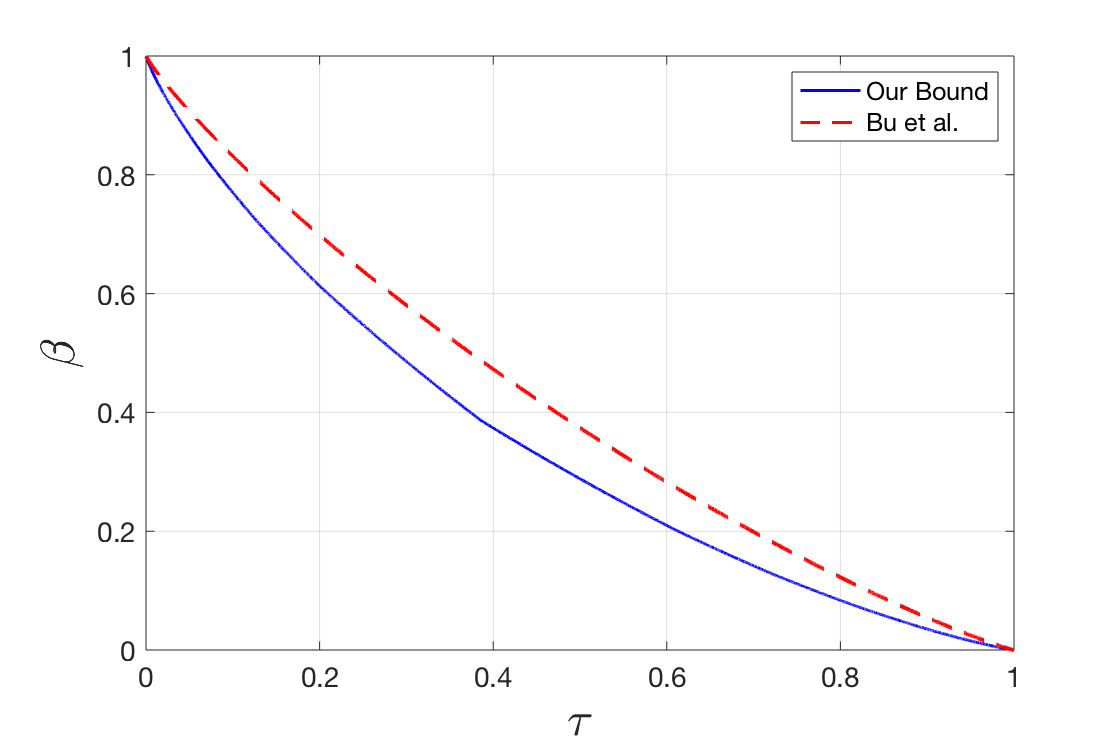}
		\caption{$\sigma = 1.3$, $Tq=30$}
	\end{subfigure}
	\caption{The outer bounds for the privacy region of SGD algorithm according to our RDP-based bound \eqref{eq:PriRegion_SGD2} (blue solid curve) and  $f$-DP  \cite{GaussianDP_deep} (red dashed curve) with the subsampling rate $q=256/60000$. As the blue curve lies above the red curve for $\sigma \leq 0.7$, our bound yields tighter privacy region.	Since the intersection in \eqref{eq:PriRegion_SGD2} is over only integer $\alpha$, the blue curve may not be smooth for large $T$.}
	\label{fig:RDP_Tradeoff_SGD_sigmaT}
\end{figure}
For a more realistic application, we apply Corollary~\ref{Corollay_Gaussian_Region} to noisy SGD algorithm (i.e., Algorithm \eqref{alg:noisySGD}). 
This algorithm can be thought of as a $T$-fold composition of Gaussian mechanism with an additional feature of subsampling (line 3 in Algorithm~\ref{alg:noisySGD}) with rate $q = \frac{m}{n}$. As before, we invoke \cite[Lemma 3]{Abadi_MomentAccountant} to obtain that each  iteration of this algorithm is approximately $(\alpha, \alpha \rho_q)$-RDP where $\rho_q = \frac{q^2}{(1-q)\sigma^2}$ for positive integer $\alpha\leq 1+\sigma^2\log\frac{1}{q\sigma}$ and $q<\frac{1}{16\sigma}$. Thus, after $T$ iterations the algorithm is $(\alpha, \alpha \rho_q T)$-RDP. Corollary~\ref{Corollay_Gaussian_Region} therefore gives 
\begin{equation}\label{eq:PriRegion_SGD2}
    \calC_{\mathsf{SGD}}(T)\subseteq \bigcap_{\alpha\in \calA}\left\{(\tau, \beta): d_\alpha(\bar\tau\|\beta)\leq \alpha\rho_q T, d_\alpha(\bar\beta\|\tau)\leq \alpha\rho_q T\right\},
\end{equation}
where $\calA$ is the set of admissible $\alpha$ indicated above. On the other hand, subsampling and composition results of $f$-DP (\cite[Theorem 4.2]{GaussianDP} and \cite[Theorem 3.2]{GaussianDP}, respectively) can be exploited to approximate (asymptotically in $T$) the tradeoff function for the  Algorithm~\ref{alg:noisySGD} and thus to construct an outer bound for the privacy region \cite{GaussianDP_deep}: 
\begin{equation}\label{eq:PriRegion_SGD}
 \calC_{\mathsf{SGD}}(T)\subseteq \left\{(\tau, \beta)\in (0, 1)^2: G_{\mu}(\tau)\leq \beta\leq 1-\tau\right\},
\end{equation}
where $\mu = q\sqrt{T\big(e^{1/\sigma^2}-1\big)}$ and 
$G_\mu(\cdot)$ was defined in \eqref{Eq:Gaussian_beta}.  In Fig. \ref{fig:RDP_Tradeoff_SGD_sigmaT}, we illustrate this bound together with  \eqref{eq:PriRegion_SGD2} for different number of iterations and $\sigma$. 
The numerical findings indicate that there always exists a $\sigma_0$ for any sub-sampling rate $q$ such that our RDP-based outer bound \eqref{eq:PriRegion_SGD2} is tighter than $f$-DP bound \eqref{eq:PriRegion_SGD} for all $\sigma\leq \sigma_0$ irrespective of the number of iterations. For instance, $\sigma_0\approx 0.7$ in Fig. \ref{fig:RDP_Tradeoff_SGD_sigmaT}, that is, \eqref{eq:PriRegion_SGD2} is tighter than \eqref{eq:PriRegion_SGD} for all $\sigma\leq 0.7$ and any number of iterations. To better support this claim, we compute the the area of the regions on the right-hand sides of \eqref{eq:PriRegion_SGD2} and \eqref{eq:PriRegion_SGD} and report the differences in Fig. \ref{fig:RDP_Tradeoff_SGD} for different values of $\sigma$ and $T$. Positive numbers indicate that the former is a smaller region, or equivalently, the outer bound in \eqref{eq:PriRegion_SGD2} is tighter than \eqref{eq:PriRegion_SGD}; thus supporting our claim. 
\begin{figure}[t]
	\centering
		\includegraphics[scale = 0.27]{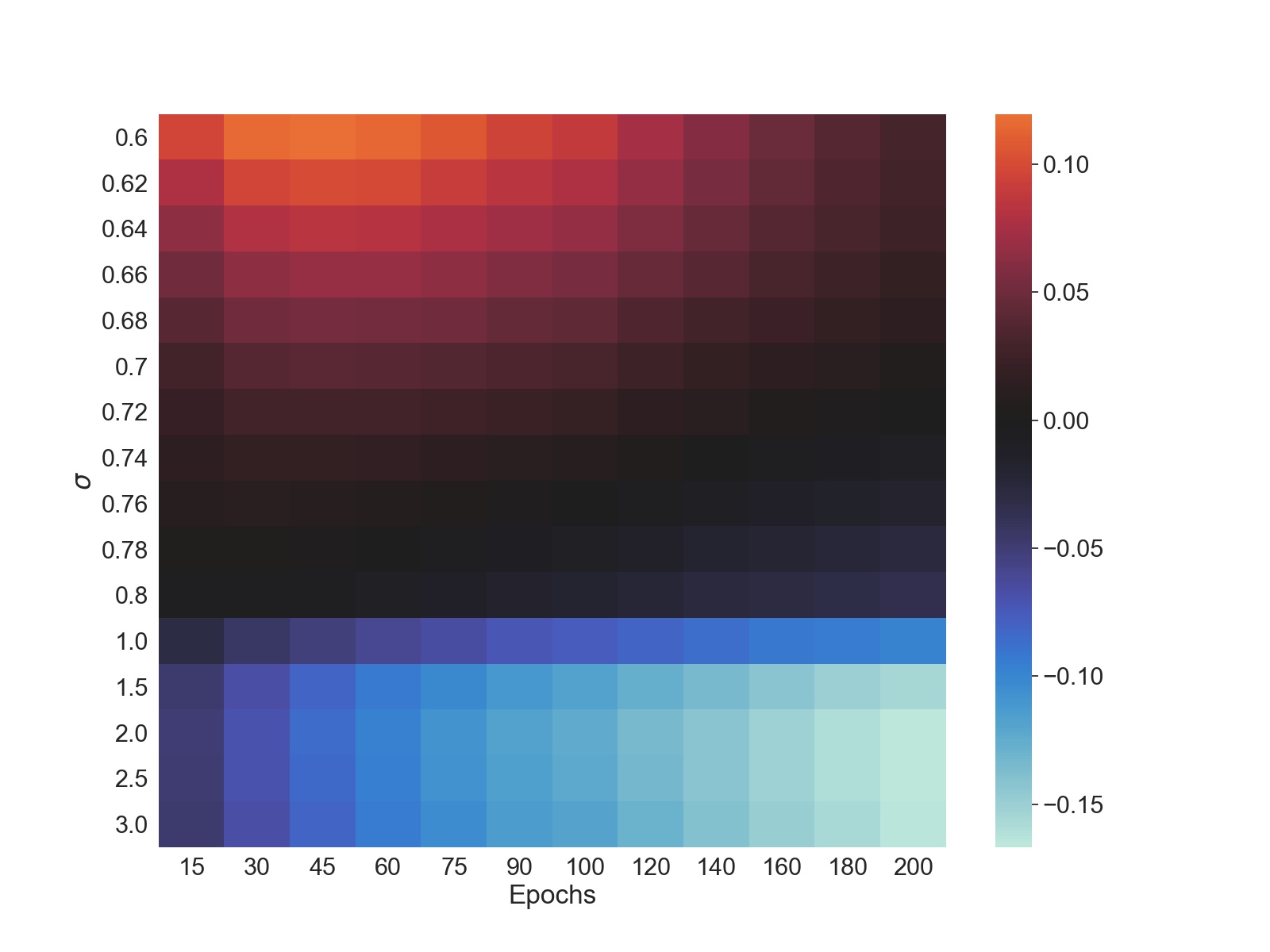}
	\caption{The difference between the area of the region in the right-hand side of \eqref{eq:PriRegion_SGD} and \eqref{eq:PriRegion_SGD2} with the subsampling rate $q=256/60000$. A positive value indicates that the outer bound in \eqref{eq:PriRegion_SGD2} is smaller than that in \eqref{eq:PriRegion_SGD}, or equivalently,  RDP leads to a tighter privacy guarantee that $f$-DP. 
	}
	\label{fig:RDP_Tradeoff_SGD}
\end{figure}

\section*{Conclusion}
In this paper, we investigated the relationship between three variants of differential privacy, namely approximate DP, R\'enyi DP, and hypothesis test DP. 
First, we established the \textit{optimal} relationship between R\'enyi DP and approximate DP that enables us to derive the optimal approximate DP parameters of a mechanism that satisfies a given level of R\'enyi DP. In order to show its practicality, we applied this result to the moments accountant framework for characterizing privacy guarantees of noisy stochastic gradient descent. When compared to the state-of-the-art, our result was shown to lead to about 100 more stochastic gradient descent iterations for training deep learning models for the same privacy budget, and thus provide better accuracy without any privacy degradation. In the second part, we analyzed  the implications of R\'enyi DP constraint in terms of the tradeoff between type I and type II error probabilities of a certain binary hypothesis test which formalizes the hypothesis test DP. More specifically, we derived an outer bound for the region of type I and type II error probabilities (also known as the \textit{privacy region}) achievable by a  mechanism that satisfies a given level of R\'enyi DP. We then used this result to characterize the privacy region of noisy stochastic gradient descent algorithm. Compared to the existing results (obtained via sub-sampling and composition results of recently proposed $f$-DP framework), our outer bound was empirically shown to be tighter for a practical range of the noise variance.

\small
\bibliography{reference}
\bibliographystyle{IEEEtran}

\appendices
\normalsize
 \section{Sufficiency of Binary Distributions for Characterizing $\mathcal R_\alpha$ }\label{Appexdix:binarycase}
 We provide a direct proof for the fact that it suffices to consider the Bernoulli distributions for characterizing $\mathcal R_\alpha$. The following argument is a natural extension of the proof of \cite[Lemma 2]{Sason_Renyi}. Let $P$ and $Q$ be two general distributions on $\mathbb X$. We wish to show that the for any $\lambda\geq 1$ and $\alpha>1$ the optimization 
 \begin{align}
     &\inf_{P, Q\in \calP(\mathbb X)} D_\alpha(P\|Q)\\
     &\qquad \text{s.t.}~ \sE_\lambda(P\|Q)\geq \delta, \nonumber
 \end{align}
 is achieved by Bernoulli distributions. 
 %Let $\calM_d$ and $\calM_{d'}$ be the output probability distributions of mechanisms $\calM$ when operating on neighboring datasets $d$ and $d'$. We wish to show that the minimization problem in \eqref{eq:Optimization_Gamma_DP} is achieved by Bernoulli distributions. 
 Let $\phi:\mathbbX\to \{1,2\}$ be defined as 
\begin{equation}
    \phi(x) = \begin{cases} 1,& \text{if}~ \frac{\text{d}P}{\text{d}Q}(x)\geq \lambda\\
2,& \text{if}~ \frac{\text{d}P}{\text{d}Q}(x)< \lambda
\end{cases}. 
\end{equation}
Also, define Bernoulli distributions $P_{\mathsf{b}}$ and $Q_{\mathsf{b}}$ on $\{1, 2\}$ as follows 
\begin{equation}
    P_{\mathsf{b}}(j) = \int_{x:\phi(x)=j}P(\text{d}x),
\end{equation}
and 
\begin{equation}
    Q_{\mathsf{b}}(j) = \int_{x:\phi(x)=j}Q(\text{d}x),
\end{equation}
for $j\in \{1, 2\}$. 
Note that in this case, we can write 
\begin{align}
\|P - \lambda Q\| &= \int_{\mathbbX}|P(\text{d}x) - \lambda Q(\text{d}x)|\\
&= \int_{\phi(x) = 1}(P(\text{d}x) - \lambda Q(\text{d}x)) \nonumber\\
&\qquad + \int_{\phi(x) = 2}(\lambda Q(\text{d}x) - P(\text{d}x))\\
&= P_{\mathsf{b}}(1) - \lambda Q_{\mathsf{b}}(1) + \lambda Q_{\mathsf{b}}(2) -  P_{\mathsf{b}}(2)\\
&= |P_{\mathsf{b}}(1) - \lambda Q_{\mathsf{b}}(1)| + |P_{\mathsf{b}}(2)-\lambda Q_{\mathsf{b}}(2)| \\
&= \|P_{\mathsf{b}} - \lambda Q_{\mathsf{b}}\|.
\end{align} 
Notice that $\sE_\lambda(P\|Q) = \frac{1}{2}\|P-\lambda Q\| +\frac{1}{2}(1-\lambda)$ and hence the above implies that 
$\sE_{\lambda}(P\|Q) = \sE_{\lambda}(P_{\mathsf{b}}\|Q_{\mathsf{b}})$. On the other hand, the data processing inequality for R\'enyi divergence implies that $D_\alpha(P\|Q)\geq D_{\alpha}(P_{\mathsf{b}}\|Q_{\mathsf{b}})$. These two observations demonstrate that the minimum of $D_\alpha(P\|Q)$ subject to  $\sE_{\lambda}(P\|Q)\geq \delta$ is achieved by Bernoulli distributions.

 \section{Proof of Theorem \ref{thm:optimization_formulation_gamma}}\label{Appendix_gamma_Optimization}
 First notice that, in light of Theorem \ref{thm:f-divergence_joint-region}, the convex set $\calR_\alpha$ defined in \eqref{ConvexSet1} is equal to the convex hull of the set $\calB_{\alpha,\eps}$ given by
 \begin{align}
 \label{eq:convex_set_RenHock_Binary}
   \calB_{\alpha,\eps}=\{(\chi^\alpha(P_{\sf b}\|Q_{\sf b}),\sEs(P_{\sf b}\|Q_{\sf b}))\big|P_{\sf b},Q_{\sf b}\in \calP(\{0,1\})\} 
 \end{align}
where $P_{\sf b}=\mathsf{Bernoulli}(p)$ and $Q_{\sf b}=\mathsf{Bernoulli}(q)$ with parameters $p,q\in(0,1)$. 
%In fact, the convex hull operation can be omitted since the set $\calB_{\alpha,\eps}$ is convex shown as follows.\\
For any pair of such distributions, define $\tilde{\gamma}\coloneqq \chi^\alpha(P_{\sf b}\|Q_{\sf b}) $ and $\delta\coloneqq \sE_{e^\eps}(P_{\sf b}\|Q_{\sf b})$.  
We first show that the convex hull of $\calB_{\alpha,\eps}$ is given by 
\begin{align}
\label{eq:convex_hull_RenHock_Binary}
  \bar{\calB}_{\alpha,\eps}=\{(\tilde{\gamma}, \delta)\big| \delta\in[0,1), \tilde{\gamma}\geq \tilde{\gamma}(\delta)\} 
\end{align}
with $\tilde{\gamma}(\delta)$ given by 
\begin{align}
	\tilde{\gamma}(\delta)= &\inf_{0<p,q<1}  \chi^\alpha(P_{\sf b}\|Q_{\sf b}) \label{eq:chracterize_Binaryset}\\
	&\qquad \text{s.t.~} \sEs(P_{\sf b}\|Q_{\sf b})\geq \delta.\nonumber
\end{align}
To this goal, we need to demonstrate that for any $\lambda\in [0,1]$ and pairs of points $(\tilde{\gamma}_1, \delta_1), (\tilde{\gamma}_2, \delta_2)\in \calB_{\alpha,\eps}$, we have  $(\lambda\tilde{\gamma}_1+\bar\lambda\tilde{\gamma}_2, \lambda\delta_1+\bar\lambda\delta_2)\in \bar{\calB}_{\alpha,\eps}$, where $\bar\lambda=1-\lambda$, or equivalently  $\lambda\delta_1+\bar\lambda\delta_2\in [0,1)$ and $\lambda\tilde{\gamma}_1+\bar\lambda\tilde{\gamma}_2\geq \tilde{\gamma}(\lambda\delta_1+\bar\lambda\delta_2)$. Hence, it suffices to show that $\delta\mapsto \tilde{\gamma}(\delta)$ is convex. 

%Therefore, the convex hull of $\calB_{\alpha,\eps}$ is given by

%To show the set $\bar{\calB}_{\alpha,\eps}$ defined in \eqref{eq:convex_hull_RenHock_Binary} is the convex  hull\footnote{Note that given $\eps\geq 0$, for any $\delta>0$ there always exists a pair of $P_{\sf b},Q_{\sf b}\in \calP(\{0,1\})$ such that $\sEs(P_{\sf b}\|Q_{\sf b})=\delta$  and $\chi^\alpha(P_{\sf b}\|Q_{\sf b})\to \infty$. If $\delta=0$, $\chi^\alpha(P_{\sf b}\|Q_{\sf b})\in [0, e^{\epsilon\alpha})$ and $\tilde{\gamma}(0)=0$. Therefore, the set $\bar{\calB}_{\alpha,\eps}$ in \eqref{eq:convex_hull_RenHock_Binary} is the smallest convex of $\calB_{\alpha,\eps}$. \SA{I don't understand the point of this footnote.}} of $\calB_{\alpha,\eps}$, we need to prove that for any $\lambda\in [0,1]$, if $(\tilde{\gamma}_1, \delta_1), (\tilde{\gamma}_2, \delta_2)\in \calB_{\alpha,\eps}$, $(\lambda\tilde{\gamma}_1+\bar\lambda\tilde{\gamma}_2, \lambda\delta_1+\bar\lambda\delta_2)\in \bar{\calB}_{\alpha,\eps}$, where $\bar\lambda=1-\lambda$, i.e., $\lambda\delta_1+\bar\lambda\delta_2\in [0,1)$ and $\lambda\tilde{\gamma}_1+\bar\lambda\tilde{\gamma}_2\geq \tilde{\gamma}(\lambda\delta_1+\bar\lambda\delta_2)$. Therefore, a sufficient condition for $\bar{\calB}_{\alpha,\eps}$ being the convex hull of $\calB_{\alpha,\eps}$ is that $\lambda\tilde{\gamma}(\delta_1)+\bar\lambda\tilde{\gamma}(\delta_2)\geq \tilde{\gamma}(\lambda\delta_1+\bar\lambda\delta_2)$, which is exactly the condition of the convexity of the function $\tilde{\gamma}(\delta)$.\\
Let $p_i,q_i\in (0,1)$ with $p_i\geq q_i$ be the optimal solution of \eqref{eq:chracterize_Binaryset} for $\delta_i$, $i=1,2$, and $P_{{\sf b},i},Q_{{\sf b},i}$ be the corresponding Bernoulli distributions. For any $\lambda\in [0,1]$, we construct two Bernoulli distribution $P_{{\sf b},\lambda}$ and $Q_{{\sf b},\lambda}$ with parameters $p_\lambda=\lambda p_1+\bar\lambda p_2$ and $q_\lambda=\lambda q_1+\bar\lambda q_2$, respectively. It can be verified that 
\begin{align}
    \sEs(P_{{\sf b},\lambda}\|Q_{{\sf b},\lambda})
    =&p_\lambda-e^\eps q_\lambda\\
    = & \lambda p_1+\bar\lambda p_2 -e^\eps(\lambda q_1+\bar\lambda q_2)\\
    \geq & \lambda \delta_1 + \bar\lambda\delta_2, 
\end{align}
i.e., $(p_\lambda, q_\lambda)$ is feasible for $\lambda\delta_1+\bar\lambda\delta_2$. In addition, from the convexity of $\chi^{\alpha}$, we have that
\begin{align}
  \lambda\tilde{\gamma}(\delta_1)+\bar\lambda\tilde{\gamma}(\delta_2) =& \lambda \chi^\alpha(P_{{\sf b},1}\|Q_{{\sf b},1})+\bar\lambda \chi^\alpha(P_{{\sf b},2}\|Q_{{\sf b},2})\\
  \geq & \chi^\alpha(P_{{\sf b},\lambda}\|Q_{{\sf b},\lambda})\\
  \geq & \tilde{\gamma}(\lambda\delta_1+\bar\lambda\delta_2).
\end{align}
Therefore, the function $\tilde{\gamma}(\delta)$ is convex in $\delta$ and hence $\bar{\calB}_{\alpha,\eps}$ is the convex hull of $\calB_{\alpha,\eps}$. In light of Theorem \ref{thm:f-divergence_joint-region}, this in turn implies that $\calR_\alpha = \bar{\calB}_{\alpha,\eps}$.
%\SA{No! You just showed that $\delta\mapsto \tilde\gamma(\delta)$ is convex! Why this convexity implies that $\conv(\calB_{\alpha, \eps}) = \bar{\calB}_{\alpha}}$}

% %\JL{This is explained in the paragraph after (31). If (30) is not the convex hull of (29), why its upper bound in (31) is the same as the upper bound of the set in (14)? In addition, the convex hull is the smallest convex set of a set and I think I have a footnote to claim (30) is the smallest convex set.}

The above analysis shows that   $\delta\mapsto \tilde{\gamma}(\delta)$ in fact constitutes the upper boundary of $\calB_{\alpha,\eps}$ and thus $\calR_\alpha$. 
%\eqref{eq:chracterize_Binaryset} as the upper bound of the convex set $\bar{\calB}_{\alpha,\eps}$, also $\calR_\alpha$. 
Since $\chi(\cdot)$ is a bijection, this allows us to deduce  
\begin{align}
    \label{eq:Optimization_Gamma_DP_Binary}
	 \gamma_\alpha^\eps(\delta) =&\inf_{0<p,q<1}  \chi^{-1}\left(\chi^\alpha(P_{\sf b}\|Q_{\sf b})\right)\\
	&\qquad \text{s.t.~} \sEs(P_{\sf b}\|Q_{\sf b})\geq \delta, \nonumber
\end{align}
and hence the optimization problem 
% From the bijective mapping $\chi(\cdot)$, the upper bound of the convex set $\calR_\alpha$ is characterized by the optimization problem in \eqref{eq:Optimization_Gamma_DP}. Similarly, we characterize the upper bound of  $\calB_{\alpha,\eps}$ in \eqref{eq:convex_set_RenHock_Binary} via
% Referring to the equivalence of the convex hull of $\calR_\alpha$ and $\calB_{\alpha,\eps}$, 
\eqref{eq:Optimization_Gamma_DP}  can be converted to the above two-parameter optimization problem.

Expanding both $\chi^\alpha$ and $\sE_{e^\eps}$, we can explicitly write \eqref{eq:Optimization_Gamma_DP_Binary} as
\begin{align}
\label{eq:Optimization_Gamma_binary} 
\gamma^\eps_\alpha(\delta)= &\inf_{0< q< p< 1}~ \frac{1}{\alpha-1}\log \left( p^\alpha q^{1-\alpha}+\bar p^\alpha\bar q^{1-\alpha}\right)\\
&\qquad \text{s.t.~} p-qe^\eps\geq\delta, \nonumber
\end{align}
where $\delta<1$ and $\gamma<\infty$. Let $h(p,q;\alpha)$ indicate the objective function of the optimization problem in \eqref{eq:Optimization_Gamma_binary}. For any given $\alpha>1$ and $ p\in (0,1)$, the partial derivative of $h(p,q;\alpha)$ with respect to $q$ is given by 
	\begin{align}
		\frac{\partial \, h(p,q;\alpha)}{\partial q} = \frac{  p^\alpha q^{-\alpha}-(1-p)^\alpha(1-q)^{-\alpha}}{ p^\alpha q^{1-\alpha}+(1-p)^\alpha(1-q)^{1-\alpha}},
	\end{align}
which is negative for all $0<q<p<1$, and therefore, $h(p,q;\alpha)$ is decreasing in $q$. In addition, for $\eps\geq 0$ and $\delta\in [0, 1)$, the two constraints $0<q<p<1$ and	$p-qe^\eps \geq \delta$ in \eqref{eq:Optimization_Gamma_binary}
can be equivalently rewritten as 
\begin{align}
	\begin{cases}
	 \delta<p<1\\
	 0<q<\frac{p-\delta}{e^\eps}.
	\end{cases}
\end{align}
Thus, the infimum in \eqref{eq:Optimization_Gamma_binary} is attained at $q=\frac{p-\delta}{e^\eps}$, and therefore, for $\alpha>1$, $\delta\in [0, 1)$ and $\eps\geq 0$, the optimization problem in \eqref{eq:Optimization_Gamma_binary} is simplified as
\begin{align}
	e^{(\alpha-1)(\gamma^\eps_\alpha(\delta)-\eps)}= &\inf_{p\in(\delta,1)}~  p^\alpha (p-\delta)^{1-\alpha}+\bar p^\alpha(e^\eps-p+\delta)^{1-\alpha}, \label{eq:Optimization_Gamma_binary_simplified}
\end{align}
which is the desired result.
 \section{Proof of Theorem \ref{Thm:Lower_BOund_Gamma}} \label{Appendix:Thm_Gamma_LB}
Recall that the optimization problem in Theorem~\ref{thm:optimization_formulation_gamma} is equivalent to \eqref{eq:Optimization_Gamma_binary_simplified}. 
Let $h_1(p;\alpha,\delta,\eps)$ indicate the objective function in \eqref{eq:Optimization_Gamma_binary_simplified}. One can verify that for $\alpha>1, \delta\in [0, 1)$ and $\eps>0$, the mapping  $p\mapsto h_1(p;\alpha,\delta,\eps)$ is convex. Therefore, the numerical result of $\gamma^\eps_\alpha(\delta)$ can be easily obtained for any given $\alpha,\delta$ and $\eps$.

To get closed-form expressions, we explore lower bounds of \eqref{eq:Optimization_Gamma_binary_simplified} as follows.\\
\textbf{Lower bound 1:} Ignoring the second term in  $h_1(p;\alpha,\delta,\eps)$, we obtain  
\begin{align}
	\label{eq:Optimization_Gamma_binary_simplified_LB1}
	e^{(\alpha-1)(\gamma^\eps_\alpha(\delta)-\eps)}\geq &\inf_{p\in(\delta,1)}~  p^\alpha (p-\delta)^{1-\alpha}
\end{align}
 We note that the objective function in \eqref{eq:Optimization_Gamma_binary_simplified_LB1} is convex in $p$, as it can be verified that $\frac{\partial^2}{\partial p^2}
	p^\alpha (p-\delta)^{1-\alpha}$ equals 
	$$(\alpha-1)\alpha\left(p^{\frac{\alpha}{2}}(p-\delta)^{\frac{-1-\alpha}{2}}-p^{\frac{\alpha-2}{2}}(p-\delta)^{\frac{1-\alpha}{2}}\right)^2\geq 0,$$
% \begin{align*}
% 	&\frac{\partial^2}{\partial p^2}
% 	p^\alpha (p-\delta)^{1-\alpha}=\\ 
% 	& \qquad \quad (\alpha-1)\alpha\left(p^{\frac{\alpha}{2}}(p-\delta)^{\frac{-1-\alpha}{2}}-p^{\frac{\alpha-2}{2}}(p-\delta)^{\frac{1-\alpha}{2}}\right)^2\geq 0,
% \end{align*}
and therefore, by setting the first derivative to be $0$,
%i.e., $\frac{{\rm d} \, \left(p^\alpha (p-\delta)^{1-\alpha}\right) }{{\rm d}\, p}=0$, 
we obtain the optimal solution for the the corresponding unconstrained problem as $p^*=\alpha\delta$. Since $\alpha>1$, it follows that the optimal solution of \eqref{eq:Optimization_Gamma_binary_simplified_LB1} is given by $p^*=\min\{\alpha\delta, 1\}$, and therefore
\begin{align}
	e^{(\alpha-1)(\gamma^\eps_\alpha(\delta)-\eps)}\geq  &  \left(\delta\alpha^\alpha (\alpha-1)^{1-\alpha}\right){\textbf{1}\{\alpha\delta< 1\}}\nonumber \\\qquad &\qquad +\left((1-\delta)^{1-\alpha}\right){\textbf{1}\{\alpha\delta\geq  1\}}
\end{align}
with equality holds if and only if $\alpha\delta\geq 1$, where $\textbf{1}\{\cdot\}$ denotes the indicator function. 
Thus, if $\alpha\delta\geq 1$, we have 
$\gamma^\eps_\alpha(\delta)=\eps-\log(1-\delta), $
and if $\alpha\delta< 1$, we have the lower bound 
\begin{align}
  \gamma^\eps_\alpha(\delta)&\geq  \eps -\frac{1}{\alpha-1}\log\left(\frac{1}{\delta\alpha}\left(1-\frac{1}{\alpha}\right)^{\alpha-1}\right)\\
  &= \eps -\frac{1}{\alpha-1}\log\frac{\zeta_\alpha}{\delta}. \label{Proof_g}
\end{align}

\noindent\textbf{Lower bound 2:} To obtain the second lower bound, we note that the function $h_1(p;\alpha,\delta,\eps)$ is convex in $\delta$. This enables us to bound $h_1(p;\alpha,\delta,\eps)$ from below by using its linear approximation at $\delta=0$. Hence we can write
%Since $h_1(p;\alpha,\delta,\eps)$ is convex in $\delta\in[0,1)$, we have
% Taking the first and second partial derivatives of $h_1(p;\alpha,\delta,\eps)$ with respect to $\delta$, we have
% \begin{align}
% 	\frac{\partial h_1(p;\alpha,\delta,\eps)}{\partial \delta}& =(\alpha-1)p^{\alpha}(p-\delta)^{-\alpha}-(\alpha-1)(1-p)^\alpha(e^\eps+\delta-p)^{-\alpha}\\
% 	\frac{\partial^2 h_1(p;\alpha,\delta,\eps)}{\partial \delta^2}& =\alpha(\alpha-1)p^{\alpha}(p-\delta)^{-\alpha-1}+\alpha(\alpha-1)(1-p)^\alpha(e^\eps+\delta-p)^{-\alpha-1}\geq 0
% \end{align}
% and therefore, 
\begin{align*}
	h_1(p;\alpha,\delta,\eps)
	&\geq  h_1(p;\alpha,\delta=0,\eps) + \frac{\partial h_1(p;\alpha,\delta=0,\eps)}{\partial \delta}\delta\\
	&=  p + (\alpha-1)\delta +\left(\frac{1-p}{e^\eps-p}\right)^\alpha\\
	&\qquad \qquad \qquad \qquad\qquad\cdot \left(e^\eps-p-(\alpha-1)\delta\right),
\end{align*}
with equality if and only if $\delta=0$.
Therefore, we have
\begin{align}
	\label{eq:Optimization_Gamma_binary_simplified_LB2}
	e^{(\alpha-1)(\gamma^\eps_\alpha(\delta)-\eps)}&\geq \inf_{p\in(\delta,1) }~  \left(1-\left(\frac{1-p}{e^\eps-p}\right)^\alpha\right)p\\ & \quad +\left(\frac{1-p}{e^\eps-p}\right)^\alpha\left(e^\eps-(\alpha-1)\delta\right) + (\alpha-1)\delta\nonumber.
\end{align}
Let $h_2(p;\alpha,\delta,\eps)$ indicate the objective function of \eqref{eq:Optimization_Gamma_binary_simplified_LB2}.
%which seems to be increasing in $a$ due to the fact that the expression $\left(\frac{1-p}{e^\eps-p}\right)^\alpha$ is less than $1$ and tends to be $0$ as $\alpha\to \infty$. 
In the following, we prove the monotonicity of $h_2(p;\alpha,\delta,\eps)$ in $p$ for $\alpha>1$, $1> \delta\geq 0$ and $\eps\geq 0$.
Taking the first derivative of $h_2(p;\alpha,\delta,\eps)$ with respect to $p$, we have
\begin{align}
	\frac{\partial \,h_2(p;\alpha,\delta,\eps) }{\partial \, p}&=1+\left(\frac{1-p}{e^\eps-p}\right)^\alpha\nonumber\\ &\qquad~ \cdot\left(\frac{\alpha(e^\eps-1)(p+(\alpha-1)\delta-e^\eps)}{(e^\eps-p)(1-p)}-1\right)\nonumber\\
	&\eqqcolon h_3(p;\alpha,\delta,\eps)\nonumber\\
	\label{eq:Gamma_bd2_InPf-1}
	&\geq 1+\left(\frac{1-p}{e^\eps-p}\right)^\alpha
	 \left(-\frac{\alpha(e^\eps-1)}{1-p}-1\right)\\
	 & \eqqcolon h_4(p;\alpha,\eps)\nonumber\\
	 \label{eq:Gamma_bd2_InPf-2}
	 &>  h_4(p=\delta;\alpha,\eps)\\
	 &= \frac{(e^\eps-\delta)^{\alpha}-\bar\delta^\alpha-\alpha(e^\eps-1)\bar\delta^{\alpha-1}}{(e^\eps-\delta)^{\alpha}}\\
	 &\eqqcolon \frac{h_5(\delta,\alpha,\eps)}{(e^\eps-\delta)^{\alpha}} \\
	 \label{eq:Gamma_bd2_InPf-3}
	 &\geq   \frac{h_5(\delta,\alpha,\eps=0)}{(e^\eps-\delta)^{\alpha}} =  0,
\end{align}
where
\begin{itemize}
	\item the inequality in \eqref{eq:Gamma_bd2_InPf-1} follows from the fact that the function $h_3(p;\alpha,\delta,\eps)$ is increasing in $\delta$, and therefore, for $1> \delta\geq 0$, $ h_3(p;\alpha,\delta,\eps)\geq h_3(p;\alpha,\delta=0,\eps) = h_4(p;\alpha,\eps)$ 
	\item the inequality in \eqref{eq:Gamma_bd2_InPf-2} is due to the fact that the function $h_4(p;\alpha,\eps)$ is increasing in $p$ as shown below  	\begin{align*}
		\frac{\partial\, h_4(p;\alpha,\eps)}{\partial p}& = \alpha(\alpha-1)(e^\eps-1)^2\bar p^{\alpha-2}(e^\eps-p)^{-\alpha-1}>0
	\end{align*}
and therefore, for $ p\in (\delta,1)$, $h_4(p;\alpha,\eps)>h_4(p=\delta;\alpha,\eps)$.
\item the inequality in \eqref{eq:Gamma_bd2_InPf-3} is from the monotonicity of the function $h_5(\delta,\alpha,\eps)$ in $\eps$. Specifically,
\begin{align*}
	\frac{\partial \, h_5(\delta,\alpha,\eps)}{\partial\, \eps} =\alpha e^\eps\left((e^\eps-\delta)^{\alpha-1} -(1-\delta)^{\alpha-1}\right)\geq 0
\end{align*}
and thus, for $\eps\geq 0$, $h_5(\delta,\alpha,\eps)\geq h_5(\delta,\alpha,\eps=0)=0$.
\end{itemize}
Therefore, the objective function $h_2(p;\alpha,\delta,\eps)$ in \eqref{eq:Optimization_Gamma_binary_simplified_LB2} is increasing in $p$, and therefore, we have
\begin{align}
		e^{(\alpha-1)(\gamma^\eps_\alpha(\delta)-\eps)}&\geq  h_2(p=\delta;\alpha,\delta,\eps)\\
		&=\alpha\delta+ \left(\frac{1-\delta}{e^\eps-\delta}\right)^\alpha\left(e^\eps-\alpha\delta\right) 
\end{align}
with equality if and only if $\delta=0$. Thus, we have 
\begin{equation}\label{Proof_f}
    \gamma^\eps_\alpha(\delta)\geq \eps+\frac{1}{\alpha-1}\log \left(\alpha\delta+ \left(\frac{1-\delta}{e^\eps-\delta}\right)^\alpha\left(e^\eps-\alpha\delta\right) \right)
\end{equation}
where the equality holds if and only if $\delta=0$ which leads to $\gamma^\eps_\alpha(\delta=0)=0$. The lower bounds \eqref{Proof_g} and \eqref{Proof_f} give the desired result. 

\section{Proof of Lemma \ref{Lemma_epsilon_Approximate}}\label{Appendix:Lemma_epsilon_Approximate}
From the first part of the proof of   Theorem~\ref{Thm:Lower_BOund_Gamma}, we have
    \begin{align}
		\eps^\delta_\alpha(\gamma)\begin{cases}
			\leq \big(\gamma-\frac{1}{\alpha-1}\log\frac{\delta}{\zeta_\alpha}\big)_{+}
			, & {\rm if  } \, \alpha\delta\leq 1\\
			= \big(\gamma +\log(1-\delta) \big)_{+}& {\rm otherwise}.
		\end{cases}
	\end{align}

% 	\begin{align}
% 		\eps^\delta_\alpha(\gamma)\begin{cases}
% 			\leq \gamma-\frac{1}{\alpha-1}\log\frac{\delta}{\zeta_\alpha}
% 			%\gamma-\frac{\log \delta}{\alpha-1}+ \log(\alpha-1)-\frac{\alpha\log\alpha}{\alpha-1}
% 			, & {\rm if  } \, \alpha\delta\leq 1\\
% 			= \gamma +\log(1-\delta) & {\rm otherwise}.
% 		\end{cases}
% 	\end{align} 
    %\SA{Jiachun, plz rewrite the above highlighted part....very confusing! For example, what do you mean by "$\zeta_\alpha e^{(\alpha-1)\gamma}\leq 1/\alpha$, i.e.,  $1-e^{-\gamma}\leq 1/\alpha$"?}\JL{I mean the two inequalities are equivalent. Now, it is rephrased as 'which can be simplified as'} \SA{can you plz quickly show it to me here?}\JL{The inequality is equivalent to $(1-1/a)^{a-1}e^{(a-1)\gamma}\leq 1$, such that $1-1/a\leq e^{-\gamma}$. Note $a=\alpha>1$} \SA{Agreed!}
    
    Next, we obtain a closed-form upper bound on $\eps^\delta_\alpha(\gamma)$ from the function $f(\alpha,\eps,\delta)$ in Theorem \ref{Thm:Lower_BOund_Gamma}. To do so, let $f_1(\alpha,\eps,\delta)$ be the expression inside the logarithm in $f(\alpha,\eps,\delta)$, i.e.,  $f_1(\alpha,\eps,\delta)\coloneqq \left(e^{\eps }-\alpha  \delta \right) \left(\frac{\delta -1}{\delta -e^{\eps }}\right)^{\alpha }+\alpha  \delta$. The second partial derivative of $f_1(\delta,\alpha,\eps)$ with respect to $\delta$ is given by
	\begin{align*}
		%\frac{\partial^2 f_1(\delta,\alpha,\eps)}{\partial\, \delta^2}&=
		(\alpha-1) \alpha \left(e^\eps-1\right)\bar\delta^\alpha \frac{\left(e^\eps \left(1-2 \delta+e^\eps\right)-\alpha \delta \left(e^\eps-1\right)\right)}{\bar\delta^2(e^\eps-\delta)^\alpha \left(\delta-e^\eps\right)^2}.
	\end{align*}
	Therefore, for $\alpha>1$, $\eps\geq 0$ and $\delta\in (0,1)$, the convexity of $f_1(\delta,\alpha,\eps)$ in $\delta$ is guaranteed by 
	\begin{align}
		\label{eq:Gamma_BD2App_condition}
		%\delta\leq \frac{e^\eps(e^\eps+1)}{2e^\eps+\alpha(e^\eps-1)}
		\delta-  \frac{e^{\eps}(e^\eps+1)}{2e^{\eps}+\alpha(e^{\eps}-1)}\leq 0 .
	\end{align}
	Let $f_2(\alpha,\eps)\coloneqq  \frac{e^{\eps}(e^\eps+1)}{2e^{\eps}+\alpha(e^{\eps}-1)} $, and therefore, if $\delta- f_2(\alpha,\eps)\leq 0$, we have
	\begin{align}
    \gamma^\eps_\alpha(\delta)
    &\geq f(\alpha,\eps,\delta)=\eps +\frac{1}{\alpha-1}\log\left(f_1(\alpha,\eps,\delta)\right)\\
		&\geq  \eps +\frac{1}{\alpha-1}\log\left(f_1(\alpha,\eps,\delta=0)+ \frac{\partial \, f_1(\delta=0,\alpha,\eps)}{\partial\, \delta}\delta\right)\nonumber\\
		&=\eps +\frac{1}{\alpha-1}\log\left(e^{-\eps(\alpha-1) }+\alpha\delta -\alpha\delta e^{-\eps(\alpha-1)}\right)\label{eq:eps_bound_inPf0},
	\end{align}
	with equality if and only if $\delta=0$. 
% 	and the corresponding upper bound of $\eps$ (as a function of $\alpha$, $\delta$ and $\gamma$) is given by 
% 	\begin{align}
% 		\label{eq:epsilon_BD2_InPf}
% 		\eps\leq \frac{1}{\alpha-1}\log\left(\frac{e^{(\alpha-1)\gamma}-1}{\alpha\delta}+1\right),
% 	\end{align}
	In the following, we prove that $\delta \leq \frac{1}{\alpha}$ is a sufficient condition for $\delta- f_2(\alpha,\eps)\leq 0$ by showing that $f_2(\alpha,\eps)>1/\alpha$ for any $\alpha>1$. Taking the first partial derivative of  $f_2(\alpha,\eps)$ with respect to $\eps$, we have
	\begin{align}
		\frac{\partial \, f_2(\alpha,\eps)}{\partial\, \eps}&= \frac{e^\eps ((2+\alpha)e^{2\eps}-2\alpha e^{2\eps}-\alpha)}{(2e^\eps+\alpha(e^\eps-1))^2}\\
		&\begin{cases}
			\leq  0, & 1\leq e^\eps \leq \frac{\alpha+\sqrt{2\alpha(\alpha+1)}}{2+\alpha} \\
			> 0, & {\rm otherwise},
		\end{cases}
	\end{align}
	and therefore, 
	\begin{align}
		f_2(\alpha,\eps)-\frac{1}{\alpha}&\geq f_2\left(\alpha,\eps=\log \frac{\alpha+\sqrt{2\alpha(\alpha+1)}}{2+\alpha} \right)-\frac{1}{\alpha}\nonumber\\
		&=  \frac{2(\alpha^2+\alpha(\sqrt{2\alpha(\alpha+1)}-1)-2)}{\alpha(2+\alpha)^2}\\
		&\eqqcolon \frac{f_3(\alpha)}{\alpha(2+\alpha)^2}\\
		\label{eq:epsilon_bd2_InPf}
		 &>  \frac{f_3(1)}{\alpha(2+\alpha)^2}=0,
	\end{align}
	where the inequality in \eqref{eq:epsilon_bd2_InPf} follows from the fact that $f_3(\alpha)$ is monotonically increasing in $\alpha>1$ as shown below: 
	\begin{align}
		\frac{{\rm d} f_3(\alpha)}{{\rm d}\alpha}=\frac{\sqrt{2}\alpha(1+2\alpha)}{\sqrt{\alpha(1+\alpha)}}+2\sqrt{2\alpha(1+\alpha)}+4\alpha-2>0.
	\end{align}
	Therefore, from the inequality in \eqref{eq:eps_bound_inPf0}, we have that for $\delta\leq 1/\alpha$,
	%the function $f(\alpha,\eps,\delta)$ in Theorem \ref{Thm:Lower_BOund_Gamma},
	\begin{align}
		\eps^\delta_\alpha(\gamma) \leq & \frac{1}{\alpha-1}\log\left(\frac{e^{(\alpha-1)\gamma}-1}{\alpha\delta}+1\right)
		%= & \frac{1}{\alpha-1}\log\left(\frac{(\alpha-1)\chi(\gamma)}{\alpha\delta}+1\right)
	\end{align}
	and equality holds if and only if $\gamma=0$, i.e., $\eps^\delta_\alpha(\gamma=0)=0$.

\section{Derivation of Remark~\ref{remark_Zero_Eps} }\label{Appendix_Remark-eps=0}
Note that it can be verified that $\gamma-\frac{1}{\alpha-1}\log\frac{\delta}{\zeta_\alpha}<0$ for $\delta>\zeta_\alpha e^{(\alpha-1)\gamma}$. Combined with $\alpha\delta\leq 1$, we therefore have  $\eps_\alpha^\delta(\gamma) = 0$ for   $\delta\in[\zeta_\alpha e^{(\alpha-1)\gamma},\frac{1}{\alpha}]$. To have a valid non-empty interval, we must have 
the condition $\zeta_\alpha e^{(\alpha-1)\gamma}<\frac{1}{\alpha}$ that is simplified to $1-e^{-\gamma}\leq \frac{1}{\alpha}$. 
A similar holds for the case $\alpha\delta>1$: we have $\gamma +\log(1-\delta)<0$ if $\delta>1-e^{-\gamma}$. Hence, $\eps^\delta_\alpha(\gamma)=0$ if $\delta>\max\{1-e^{-\gamma},1/\alpha\}$. 
    
\section{Proof of Lemma \ref{Lemma:Bound_on_Eps_MA}}
% We use the result in Lemma \ref{Lemma_epsilon_Approximate} to solve for the privacy parameter $\eps^\delta(\rho, T)$ for Gaussian composition where $\gamma=\alpha \rho T$ ($\rho>0$).
Recall that for the $T$-fold composition of Gaussian mechanism with variance $\sigma^2$, we have  $\gamma(\alpha)=\alpha \rho T$ where $\rho=1/{\sigma^2}$.
From Lemma \ref{Lemma_epsilon_Approximate}, we have that for $\alpha\delta\geq 1$ and $0<\delta< 1$,
\begin{align}
    \eps_\alpha^{\delta}(\rho\alpha T)=\left(\rho\alpha T +\log(1-\delta)\right)_{+}
\end{align}
and therefore, 
\begin{align}
    \eps^\delta(\rho, T) &= \inf_{\alpha>1}\eps_\alpha^{\delta}(\rho\alpha T)\\
    &\leq \inf_{\alpha\geq \frac{1}{\delta}}\left(\rho\alpha T +\log(1-\delta)\right)_{+}\\
    & = \left(\frac{\rho T}{\delta} +\log(1-\delta)\right)_{+}\label{eq:GaussianComp_InPf_alphadelta>1}.
\end{align}
%Note that if $\delta> \delta_0$ such that $\delta_0\log(1-\delta_0)=-\rho T$, the expression in \eqref{eq:GaussianComp_InPf_alphadelta>1} is negative, and therefore, it is sufficient that $\eps^\delta(\rho, T)=0$. Therefore, we have $\eps_\alpha^{\delta}(\rho\alpha T)\leq \left(\frac{\rho T}{\delta} +\log(1-\delta)\right)_{+}$ for $\alpha\delta\geq 1$.
In addition, from  Lemma \ref{Lemma_epsilon_Approximate}, we have that for $0<\alpha\delta< 1$,
	\begin{align*}
	    \eps^\delta_\alpha(\alpha\rho T)&\leq \min\Big\{\Big(\alpha\rho T-\frac{1}{\alpha-1}\log\frac{\delta}{\zeta_\alpha}\Big)_{+}\\ &\qquad\qquad \qquad  ,\frac{1}{\alpha-1}\log\Big(\frac{(\alpha-1)\chi(\alpha\rho T)}{\alpha\delta}+1\Big)  \Big\},
	\end{align*}
where $\chi(\alpha\rho T)=\frac{e^{\rho\alpha(\alpha-1)T}-1}{\alpha-1}$, and therefore, 
\begin{align}
    \eps^\delta(\rho, T) =& \inf_{\alpha>1}\eps_\alpha^{\delta}(\rho\alpha T)\nonumber\\
    \leq &\inf_{1<\alpha< \frac{1}{\delta}}\min\Big\{\Big(\alpha\rho T-\frac{1}{\alpha-1}\log\frac{\delta}{\zeta_\alpha}\Big)_+\label{eq:GaussianComp_InPf_alphadelta<1}\\ &\qquad \qquad~~ ,\frac{1}{\alpha-1}\log\Big(1+\frac{e^{\rho\alpha(\alpha-1)T}-1}{\alpha\delta}\Big)  \Big\}.\nonumber
\end{align}

Combining the two inequalities in \eqref{eq:GaussianComp_InPf_alphadelta>1} and \eqref{eq:GaussianComp_InPf_alphadelta<1}, we obtain the upper bound of $\eps^\delta(\rho, T)$ in Lemma \ref{Lemma:Bound_on_Eps_MA}.

\section{Proof of Theorem \ref{Theorem_noise}}
% From Lemma \ref{Lemma_epsilon_Approximate}, we know that for any mechanism satisfying $(\alpha,\gamma)$-RDP, it satisfies $(\eps^\delta_\alpha(\gamma),\delta)$-DP. For $T$-fold Gaussin composition, we have $\gamma=\frac{\alpha T}{2\sigma^2}$ such that
% \begin{align}
%   \eps^\delta_\alpha(\gamma)=\eps^\delta_\alpha\left(\frac{\alpha T}{2\sigma^2}\right)\leq \frac{\alpha T}{2\sigma^2}- \frac{1}{\alpha-1}\log\frac{\delta}{\zeta_\alpha} 
% \end{align}
Lemma \ref{Lemma:Bound_on_Eps_MA} illustrates that the $T$-fold adaptive homogeneous composition of the Gaussian mechanism with variance $\sigma^2$ is $(\eps, \delta)$-DP where 
\begin{align} \label{Variance_1}
  \eps =  \inf_{1<\alpha\leq  \frac{1}{\delta}}\frac{\alpha T}{2\sigma^2}- \frac{1}{\alpha-1}\log\frac{\delta}{\zeta_\alpha}.
\end{align}
% Rearrenging the above, we obtain 
% \begin{align} \label{Variance_2}\sigma^2 = \inf_{1<\alpha\leq  \frac{1}{\delta}}\frac{\alpha T}{2\eps+ \frac{2}{\alpha-1}\log\frac{\delta}{\zeta_\alpha} }
% \end{align}
Assuming that $\frac{2\log \delta^{-1}}{\eps}\leq \frac{1}{\delta}$, or equivalently $\eps\geq 2\delta\log \delta^{-1}$, we can plug $\alpha = \frac{2\log \delta^{-1}}{\eps}$ in the above  expression to derive the following lower lower bound for $\sigma^2$ 
% That is, for all
% $\alpha\in (1, 1/\delta]$  
% \begin{align}
%     \eps\leq \frac{\alpha T}{2\sigma^2}- \frac{1}{\alpha-1}\log\frac{\delta}{\zeta_\alpha}
% \end{align}
% or equivalently, 
% %let $\eps = \eps^\delta\left(\frac{1}{2\sigma^2},T\right)$.
% \begin{align}
% \sigma^2\geq \frac{\alpha T}{2\eps+ \frac{2}{\alpha-1}\log\frac{\delta}{\zeta_\alpha} }.
% \end{align}
% Therefore, if the variance of a Gaussian mechanism is no less than
% \begin{align}
% \label{eq:sigma_GaussionComp_InPf0}
%     \inf_{1<\alpha\leq \frac{1}{\delta}}\frac{\alpha T}{2\eps+ \frac{2}{\alpha-1}\log\frac{\delta}{\zeta_\alpha} },
% \end{align}
%  the $T$-fold adaptive homogeneous composition of the Gaussian mechanism satisfies $(\eps, \delta )$-DP.
%  Thus, it is sufficient to bound the required variance from below by any objective value of \eqref{eq:sigma_GaussionComp_InPf0} for $\alpha\in (1,1/\delta]$.
% Assume that $(2\log \delta^{-1})/\eps\leq 1/\delta$, or equivalently, $\eps\geq 2\delta\log \delta^{-1}$. 
% Note that the constraint $\eps\geq 2\delta\log \delta^{-1}$  holds for any $\epsilon> 0$ for sufficiently small $\delta$. Let $\alpha=\frac{2\log \delta^{-1}}{\eps}$ such that the corresponding objective value of \eqref{eq:sigma_GaussionComp_InPf0} is given by 
%Therefore, the optimal value in \eqref{eq:sigma_GaussionComp_InPf0} is bounded from above by the objective value for $\alpha=\frac{2\log \delta^{-1}}{\eps}$.
\begin{align}
    % & \inf_{1<\alpha\leq \frac{1}{\delta}}\frac{\alpha T}{2\eps+ \frac{2}{\alpha-1}\log\frac{\delta}{\zeta_\alpha} }\\
    % =&  \inf_{1<\alpha\leq \frac{1}{\delta}}\frac{\alpha T}{2\eps+ \frac{2}{\alpha-1}\log\frac{\delta}{\frac{1}{\alpha}\left(1-\frac{1}{\alpha}\right)^{\alpha-1}} }\label{eq:sigma_bd_InPf0}\\
    %& \frac{\alpha T}{2\eps+ \frac{2}{\alpha-1}\log\alpha\delta- 2\log\left(1-\frac{1}{\alpha}\right) }\bigg|_{\alpha=\frac{2\log \delta^{-1}}{\eps}}\label{eq:sigma_bd_InPf1}\\
    %  = &\frac{ (\eps-2 \log\frac{1}{\delta})T\log\frac{1}{\delta}}{\eps \left(\eps (\eps-\log\frac{1}{\delta})+(-\eps+2 \log\frac{1}{\delta}) \log \left(\frac{-\eps+2 \log\frac{1}{\delta}}{2 \log\frac{1}{\delta}}\right)-\eps \log \left(\frac{2 \log\frac{1}{\delta}}{\eps}\right)\right)}\\
    & \frac{ (\eps-2 \log\frac{1}{\delta})T\log\frac{1}{\delta}}{\eps^2 \left(\eps-\log\frac{1}{\delta}+\frac{-\eps+2 \log\frac{1}{\delta}}{\eps} \log \left(\frac{-\eps+2 \log\frac{1}{\delta}}{2 \log\frac{1}{\delta}}\right)- \log \left(\frac{2 \log\frac{1}{\delta}}{\eps}\right)\right)}\label{eq:sigma_bd_InPf2}\\
    & =  \frac{2T \log\frac{1}{\delta} }{\eps^2}+\frac{T}{\eps}-\frac{2T \left( \log \left(2 \log\frac{1}{\delta}\right)+1-\log \eps\right)}{\eps^2}+\frac{T}{2\eps^2 \log\frac{1}{\delta}}\nonumber\\
    &\cdot\Big[4 \log ^2A+(8-6 \eps) \log A+2 \eps^2-5 \eps+4\Big]+O\Big(\frac{1}{\log^2\frac{1}{\delta}}\Big)\label{eq:sigma_bd_InPf3}\\
    &=  \frac{2T}{\eps^2}\log\frac{1}{\delta} + \frac{T}{\eps} -\frac{2T}{\eps^2}\left(\log(2\log\delta^{-1})+1-\log\eps\right)\nonumber\\ 
    &\quad+ O\left(\frac{\log^2(\log\delta^{-1})}{\log\delta^{-1}}\right) \label{eq:sigma_bd_InPf4}.
\end{align}
where 
\begin{itemize}
    %\item the expression in \eqref{eq:sigma_bd_InPf1} is from the expression of $\zeta_\alpha= \frac{1}{\alpha}\left(1-\frac{1}{\alpha}\right)^{\alpha-1}$ (defined in Theorem \ref{Thm:Lower_BOund_Gamma})
    %and the condition $\eps>2\delta\log\delta^{-1}$,
    \item the expression in \eqref{eq:sigma_bd_InPf3} is the Taylor expansion of \eqref{eq:sigma_bd_InPf2} at $\delta=0$ and $A\coloneqq \frac{\log\frac{1}{\delta^2}}{\eps}$, 
    \item in \eqref{eq:sigma_bd_InPf3} as $\delta\to 0$, we have $\log \delta^{-1}\to \infty$, therefore, for any fixed finite $\eps$ and  $T$, the fourth term is of order %$\frac{\log^2(\log\delta^{-1})}{\log\delta^{-1}}$, i.e., 
    $O\left(\frac{\log^2(\log\delta^{-1})}{\log\delta^{-1}}\right)$ and dominates $O\left(\frac{1}{\log^2\delta^{-1}}\right)$. 
    %The \eqref{eq:sigma_bd_InPf3} results from keeping the larger order.
\end{itemize}
% Therefore, as the variance $\sigma^2$ of a Gaussian mechanism satisfies
% \begin{align}
%     \sigma^2\geq &\frac{2T}{\eps^2}\log\frac{1}{\delta} + \frac{T}{\eps} -\frac{2T}{\eps^2}\left(\log(2\log\delta^{-1})+1-\log\eps\right)\\
% &\quad + O\left(\frac{\log^2(\log\delta^{-1})}{\log\delta^{-1}}\right),
% \end{align}
% its $T$-fold homogeneous composition satisfies $(\eps,\delta)$-DP.
It is worth mentioning that this choice of $\alpha$ has already appeared in literature, see e.g., \cite[Discussion after Thm 35]{Feldman2018PrivacyAB}.

\section{Proof of Proposition~\ref{Prop:Polyankiy_RDP}}
\label{Appendix:Polyankiy_RDP}
Recall that $\calM_d$ and $\calM_{d'}$ are the output distributions of mechanism $\calM$ when running on two neighboring datasets $d$ and $d'$, respectively. 
For notational simplicity, let $P$ and $Q$ denote $\calM_d$ and $\calM_{d'}$, respectively and also $\beta(\tau)$ denote $\beta_\calM^{dd'}(\tau)$.
We wish to prove a more general result than Proposition~\ref{Prop:Polyankiy_RDP}: For any convex real-valued function $f$ with $f(1) = 0$, we show 
\begin{equation}\label{Eq:f_DP_beta}
    D_f(P\|Q) = \int_{0}^1|\beta'(\tau)|f\left(\frac{1}{|\beta'(\tau)|}\right)\text{d}\tau,
\end{equation}
where $\beta'(\tau) \coloneqq \frac{\text{d}}{\text{d}\tau}\beta(\tau)$. 
For a given $\lambda\geq 0$, define  
\begin{eqnarray}
\tau_\lambda\coloneqq P\Big(\frac{\text{d}P}{\text{d}Q}\leq \lambda\Big) = \int \textbf{1}\Big\{\frac{\textnormal{d}P}{\textnormal{d}Q}\leq \lambda\Big\}\text{d}P,
\end{eqnarray}
where, as before, $\textbf{1}\{\cdot\}$ denotes the indicator function.
Then, since $\tau\mapsto \beta(\tau)$ specifies the optimal tradeoff between type I and type II error probabilities of testing $P$ against $Q$, we have from Neyman-Pearson lemma that
\begin{eqnarray}
\beta(\tau_\lambda) =  Q\Big(\frac{\text{d}P}{\text{d}Q}> \lambda\Big).
\end{eqnarray}
%Let $$\calD_\lambda(P\|Q) \coloneqq \int(\text{d}P-\lambda\text{d}Q)_+.$$
%Note that $\calD_\lambda(P\|Q)$ coincides with the $\sE_\lambda$-divergence $\sE_\lambda(P\|Q)$ if $\lambda\geq 1$.

Before we begin the proof of \eqref{Eq:f_DP_beta}, we need the following fact that will be needed later. 

\noindent\textbf{Fact.}
We have \begin{equation}\label{Derivative_tau_beta}
    \frac{\text{d}}{\text{d}\lambda}\tau_\lambda = -\lambda \frac{\text{d}}{\text{d}\lambda}\beta(\tau_\lambda).
\end{equation}
We prove this fact as follows:
\begin{align}
    \bar\tau_\lambda - \lambda \beta(\tau_\lambda)  
  &= 1-\int \left[\frac{\textnormal{d}P}{\textnormal{d}Q}\textbf{1}\Big\{\frac{\textnormal{d}P}{\textnormal{d}Q}\leq \lambda\Big\} + \lambda \textbf{1}\Big\{\frac{\textnormal{d}P}{\textnormal{d}Q}> \lambda\Big\}\right]\text{d}Q\label{LHS_E_Gamma}\\
  & = 1- \int\min\left\{\frac{\textnormal{d}P}{\textnormal{d}Q}, \lambda\right\}\textnormal{d}Q\\
  %& = 1- \int\min\left\{\lambda\frac{\textnormal{d}Q}{\textnormal{d}P}, 1\right\}\textnormal{d}P \\
  & = 1- \int_{0}^1 P\left[\frac{\text{d}Q}{\text{d}P}\geq \frac{t}{\lambda}\right]\text{d}t \label{Int_Expectation}\\
  %&  = \int_{0}^1 P\left[\frac{\text{d}P}{\text{d}Q}> \frac{\lambda}{t}\right]\text{d}t \\
   & =\lambda\int_{\lambda}^\infty \frac{1}{t^2}P\left[\frac{\text{d}P}{\text{d}Q}> t\right]\text{d}t \label{Equivalent_E_Lambda}
\end{align}
where equality in \eqref{Int_Expectation} comes from the formula that $\mathbb E[U] = \int\Pr(U\geq t)\text{d}t$ for any non-negative random variable $U$. We can hence write 
\begin{equation}
    \frac{1-\tau_\lambda - \lambda\beta(\tau_\lambda)}{\lambda} = \int_{\lambda}^\infty \frac{1}{t^2}P\left[\frac{\text{d}P}{\text{d}Q}> t\right]\text{d}t.
\end{equation}
Taking a derivative, with respect to $\lambda$, of both sides of this identity, we obtain the desired result \eqref{Derivative_tau_beta}. It is worth noting that if we consider $\sE_\lambda$-divergence for any non-negative $\lambda$ (rather than $\lambda\geq 1$), then the left-hand side of \eqref{LHS_E_Gamma} is in fact equal to $\sE_\lambda(P\|Q)$, because it can be easily verified that 
\begin{align*}
    \sE_\lambda(P\|Q) & = \sup_{A\subset \mathbb X} P(A) - \lambda Q(A)
    %& = P\left(\frac{\text{d}P}{\text{d}Q}> \lambda\right) - \lambda Q\left(\frac{\text{d}P}{\text{d}Q}> \lambda\right)\\
    = \bar\tau_\lambda - \lambda \beta(\tau_\lambda).
\end{align*}
Hence, \eqref{Equivalent_E_Lambda} gives an equivalent formula for $\sE_\lambda$-divergence for $\lambda\geq 0$. 
\begin{proof}[Proof of \eqref{Eq:f_DP_beta}]
We have
\begin{align}
    D_f(P\|Q) 
    %& = \int f\left(\frac{\text{d}P}{\text{d}Q}\right)\text{d}Q\\ 
    & = \int f\left(\frac{\text{d}P}{\text{d}Q}\right)\frac{\text{d}Q}{\text{d}P}\text{d}P = \int_0^\infty f(t)\frac{1}{t}\text{d}\tau_t\\
    & = -\int_0^\infty f\Big(-\frac{\text{d}\tau_t/\text{d}t}{\text{d}\beta(\tau_t)/\text{d}t}\Big)\frac{\text{d}\beta(\tau_t)/\text{d}t}{\text{d}\tau_t/\text{d}t}\text{d}\tau_t\label{change_mesure}\\
    & = -\int_0^1 f\Big(-\frac{1}{\beta'(\tau)}\Big)\beta'(\tau)\text{d}\tau
\end{align}
where the equality in \eqref{change_mesure} follows from \eqref{Derivative_tau_beta}. The desired result follows by noticing that $\tau\mapsto \beta(\tau)$ is decreasing and hence $\beta'(\tau)\leq 0$ implying that $-\beta'(\tau)  = |\beta'(\tau)|$.
\end{proof}
Plugging $f(t)= \frac{t^\alpha-1}{\alpha-1}$ for some $\alpha>1$ into \eqref{Eq:f_DP_beta}, we obtain 
\begin{equation}
    \chi^\alpha(P\|Q) = \frac{1}{\alpha-1}\left[-\beta(0) + \int_{0}^1 |\beta'(\tau)|^{1-\alpha}\text{d}\tau  \right],
\end{equation}
implying 
\begin{align*}
    D_\alpha(P\|Q) %\frac{1}{\alpha-1}\log(1+(\alpha-1)\chi^\alpha(P\|Q)) \nonumber\\
    & = \frac{1}{\alpha-1}\log\Big(1-\beta(0)+\int_{0}^1 (|\beta'(\tau)|)^{1-\alpha}\text{d}\tau  \Big).
\end{align*}

\end{document}

%% file: preamble.tex
\usepackage[usenames,dvipsnames]{xcolor}
\usepackage{dsfont}
\usepackage{amsbsy}
\usepackage{amssymb, commath}
\usepackage{amsmath, mathtools, amsmath}%
\usepackage{amsthm}
\usepackage{enumitem}

\usepackage{breqn}
\usepackage{graphicx}
\usepackage{mathrsfs}
\usepackage{epstopdf}
\usepackage[font=small,labelsep=space]{caption}
\DeclareCaptionLabelSeparator{dot}{.~}
\usepackage{subcaption}
\usepackage[square,numbers, sort&compress]{natbib}
\usepackage{comment}
\usepackage{diagbox}
\usepackage{algorithmic}
\usepackage[ruled]{algorithm} 

\newtheorem{thm}{Theorem}
\newtheorem{lem}{Lemma}

\newtheorem{prop}{Proposition}

\newtheorem{cor}{Corollary}

\makeatletter
\newtheorem*{rep@theorem}{\rep@title}
\newcommand{\newreptheorem}[2]{%
\newenvironment{rep#1}[1]{%
 \def\rep@title{#2 \ref{##1}}%
 \begin{rep@theorem}}%
 {\end{rep@theorem}}}
\makeatother
\newreptheorem{conj}{Conjecture}

\theoremstyle{definition}

\newtheorem{defn}{Definition}

\newtheorem{remark}{Remark}

\newcommand{\mathbbX}{\mathbb{X}}
\newcommand{\mathbbD}{\mathbb{D}}

\newcommand{\calR}{\mathcal{R}}
\newcommand{\calA}{\mathcal{A}}
\newcommand{\calB}{\mathcal{B}}

\newcommand{\calM}{\mathcal{M}}

\newcommand{\calC}{\mathcal{C}}

\newcommand{\calF}{\mathcal{F}}

\newcommand{\calP}{\mathcal{P}}

\newcommand{\conv}{\mathsf{conv}}
\newcommand{\sE}{\mathsf{E}}

\renewcommand{\tilde}{\widetilde}

\newcommand{\Reals}{\mathbb{R}}

\newcommand{\eps}{\varepsilon}

\newcommand{\calN}{\mathcal{N}}

\newcommand{\sEs}{\sE_{e^\eps}}

\usepackage[colorlinks=true]{hyperref}
\hypersetup{%
,urlcolor=black
,citecolor=red
,linkcolor=blue
}

\newcommand{\eqn}[2]{\begin{equation}
\label{#1}
#2
\end{equation}}

\definecolor{light-gray}{gray}{.90}

\makeatletter
\newcommand*{\addFileDependency}[1]{% argument=file name and extension
  \typeout{(#1)}
  \@addtofilelist{#1}
  \IfFileExists{#1}{}{\typeout{No file #1.}}
}
\makeatother